\pdfoutput=1
\documentclass[12pt]{article}

\usepackage{amsfonts,amsmath,amsthm,color,graphicx,MnSymbol,esint}
\usepackage{hyperref}
\usepackage{datetime}
\usepackage{url}
\usepackage{color}
\usepackage{mathrsfs}
\usepackage{cite}
\usepackage[latin1]{inputenc}

\usepackage[OT2,T1]{fontenc}

\newcommand\encadremath[1]{\vbox{\hrule\hbox{\vrule\kern8pt 
\vbox{\kern8pt \hbox{$\displaystyle #1$}\kern8pt} 
\kern8pt\vrule}\hrule}}
\def\enca#1{\vbox{\hrule\hbox{
\vrule\kern8pt\vbox{\kern8pt \hbox{$\displaystyle #1$}
\kern8pt} \kern8pt\vrule}\hrule}}

\newcommand\figureframex[3]{
\begin{figure}[bth]
\hrule\hbox{\vrule\kern8pt 
\vbox{\kern8pt \vbox{
\begin{center}
{\mbox{\epsfxsize=#1.truecm\epsfbox{#2}}}
\end{center}
\caption{#3}
}\kern8pt} 
\kern8pt\vrule}\hrule
\end{figure}
}
\newcommand\figureframey[3]{
\begin{figure}[bth]
\hrule\hbox{\vrule\kern8pt 
\vbox{\kern8pt \vbox{
\begin{center}
{\mbox{\epsfysize=#1.truecm\epsfbox{#2}}}
\end{center}
\caption{#3}
}\kern8pt} 
\kern8pt\vrule}\hrule
\end{figure}
}

\makeatletter
\@addtoreset{equation}{section}
\makeatother
\newtheorem{theorem}{Theorem}[section]

\newtheorem{remark}{Remark}[section]
\newtheorem{proposition}{Proposition}[section]
\newtheorem{lemma}{Lemma}[section]
\newtheorem{corollary}{Corollary}[section]
\newtheorem{definition}{Definition}[section]
\def\br{\begin{remark}\rm\small}
\def\er{\end{remark}}
\def\bt{\begin{theorem}}
\def\et{\end{theorem}}
\def\bd{\begin{definition}}
\def\ed{\end{definition}}
\def\bp{\begin{proposition}}
\def\ep{\end{proposition}}
\def\bl{\begin{lemma}}
\def\el{\end{lemma}}
\def\bc{\begin{corollary}}
\def\ec{\end{corollary}}
\def\beaq{\begin{eqnarray}}
\def\eeaq{\end{eqnarray}}

\newcommand{\beq}{\begin{equation}}
\newcommand{\eeq}{\end{equation}}
\newcommand{\bea}{\begin{eqnarray}}
\newcommand{\eea}{\end{eqnarray}}

\renewcommand{\and}{{\qquad {\rm and} \qquad}}

\newcommand{\Pint}{{\int\kern -1.em -\kern-.25em}} 
\newcommand{\Vol}{\mathrm{Vol}}

\begin{document}

\sloppy

\pagestyle{empty}
\hfill IPHT T17/006\\
\indent \hfill CRM-2016-nnnn
\addtolength{\baselineskip}{0.20\baselineskip}
\baselineskip 16pt 
\begin{center}
\begin{Large}\fontfamily{cmss}
\fontsize{20pt}{30pt}
\selectfont
\medskip
\textbf{Local properties of the random Delaunay triangulation model and topological models of 2D gravity}
\\ 
\bigskip
\bigskip
\end{Large}
\bigskip
\bigskip
{\sl S\'everin Charbonnier}\hspace*{0.05cm}${}^2$,
{\sl Fran\c cois David}\hspace*{0.05cm}${}^1$, {\sl Bertrand\ Eynard}\hspace*{0.05cm}${}^{2,\,3}$
\\
\vspace{6pt}

${}^1$ Institut de Physique Th\'eorique,\\
CNRS, URA 2306, F-91191 Gif-sur-Yvette, France\\
CEA, IPhT, F-91191 Gif-sur-Yvette, France\\
\vspace{6pt}
${}^2$
Institut de Physique Th\'eorique,\\
CEA, IPhT, F-91191 Gif-sur-Yvette, France\\
CNRS, URA 2306, F-91191 Gif-sur-Yvette, France\\
\vspace{6pt}
${}^3$ CRM, Centre de Recherche Math\'ematiques, Montr\'eal QC Canada.\\
\end{center}

\vspace{20pt}
\noindent{\bf Abstract:}

Delaunay triangulations provide a bijection between a set of $N+3$ points  in generic position in the complex plane, and the set of triangulations with given circumcircle intersection angles.
The uniform Lebesgue measure on these angles translates into a K\"ahler measure for Delaunay triangulations, or equivalently on the moduli space $\mathcal M_{0,N+3}$ of genus zero Riemann surfaces with $N+3$ marked points.
We study the properties of this measure. First we relate it to the topological Weil-Petersson symplectic form on the moduli space $\mathcal M_{0,N+3}$.
Then we show that this measure, properly extended to the space of all triangulations on the plane, has maximality properties for Delaunay triangulations. Finally we show, using new local inequalities on the measures, that the volume $\mathcal{V}_N$ on triangulations with $N+3$ points is monotonically increasing when a point is added,  $N\to N+1$. 
We expect that this can be a step towards seeing that the large $N$ limit of random triangulations can tend to the Liouville conformal field theory.

\vspace{26pt}
\begin{center}
August 14, 2018
\end{center}
\pagestyle{plain}
\setcounter{page}{1}

\newpage
\tableofcontents
\newpage

\section{Introduction}
\newcommand{\EofT}{\mathcal{E}(T)}
\newcommand{\VofT}{\mathcal{V}(T)}
\newcommand{\FofT}{\mathcal{F}(T)}
\subsection{Framework}
It has been argued by theoretical physicists
\cite{David1984}\cite{Frohlich1985}\cite{Kazakov1985} that the continuous limit of large planar maps should be the same thing as two dimensional (2D) quantum gravity, i.e. a theory of random Riemannian metrics (for general references on this subject of discrete and continuous Quantum Geometry and Quantum Gravity see e.g. \cite{AmbjornDurhuusJonsson:2005}). 
In a celebrated paper, Polyakov \cite{Polyakov1981207} had already shown that (in the framework of non-critical string theory) continuum 2D quantum gravity can be reformulated in the  conformal gauge as a 2D integrable conformal field theory, the quantum Liouville theory. Together with Kniznik and Zamolodchikov, he showed later that the scaling dimensions of its local operators are encoded into the so called KPZ relations.
\cite{KPZ:1988}   \cite{David1988-F}  \cite{David1988-M}\cite{DistlerKawai1989}.
Going back to the discrete case, planar maps have been studied since decades by combinatorial and random matrix methods, in particular recursion relations.
Liouville theory has been widely studied by the technics of integrable systems and conformal field theory.
The many explicit results thus obtained corroborate the equivalence beetween the continuum limit (large maps) of planar map models and quantum Liouville theory.

Another approach, initiated  by Witten \cite{Witten:1990}, is to formulate 2D quantum gravity as a topological theory (topological gravity). This leads to the 2D topological Witten-Kontsevich intersection theory, notably studied by Kontsevich \cite{Kontsevich1992}. In this topological theories, the topological observables obey also recursion relations, and can be studied directly by matrix model technics (the Kontsevich model, the Penner model, etc.) with no need to take a continuum limit of large maps (since the theory is topological). 
Again, many explicit results corroborate the relations between 2d topological gravity, Liouville/CFT theory and 2d (non-topological) gravity  (see e.g.\cite{KostovPonsotSerban2004}\cite{ChekhovEynardRibaut2013}).
For some recent results, with a first construction of a ``large Strebel graph limit'' of a topological gravity, which shows that is coincide with the expected 2d CFT gravity, see \cite{CharbonierEynardDavid2017}. 

Random maps are now extensively studied by mathematicians. It has been shown recently, by combinations of combinatorics and probabilistics methods, that the continuous limit of large planar maps equipped with the graph distance, and thus viewed as metric spaces, exists (with the topology induces by the Gromov-Hausdorff distance between metric spaces), and converges towards the so-called ``Brownian map''
\cite{LeGall2013}\cite{Miermont2013}
(see the references therein for previous litterature).
Despite numerous ongoing efforts, the problem which has so far remained elusive is to prove the general equivalence of this limit (in the Gromov-Hausdorff topology on the ``abstract'' space of metric spaces) with the Liouville CFT in the plane, which makes reference to an explicit conformal embedding of 2D metric in the plane, via the uniformization theorem.
Tackling this problem requires methods of embedding planar maps into the Euclidian plane.
In the simple case of planar triangulations, many methods are available. 
Let us quote the "barycentric" Tutte embedding (see e.g. \cite{AmbjornBarkleyBudd2012}), and the "Regge" embedding (see e.g. sect. 6 of \cite{Hamber2009}), which are not conformal in any sense.
The exact uniformization embedding is fully conformal, but difficult to study.
The ``circle packing''  methods (see e.g. \cite{Benjamini2010}) have some interesting conformal properties.

\subsection{Summary of the model and the results}
This study deals with a very natural extension of the circle packing and circle pattern methods, introduced by two of the authors in \cite{DavidEynard2014}.
This embedding relies on the patterns of circumcircles of Delaunay triangulations.
Using the fact that the whole (moduli) space of surfaces is obtained by varying circumcircle intersection angles, they showed that the uniform measure on random planar maps, equipped with the uniform Lebesgue measure on edge angles variables, gets transported by the circle pattern embedding method, to a conformally invariant spatial point process (measure on point distributions) in the plane with many interesting properties: (i)
it has an explicit representation in term of geometrical objects (3-rooted trees) on Delaunay triangulations; 
(ii) it is a K\"ahler metric whose prepotential has a simple formulation in term of hyperbolic geometry;
(iii) it can be written as a ``discrete Fadeev-Popov'' determinant, very similar to the  conformal gauge fixing Fadeev-Popov determinant of Polyakov; (iv) it can also be written locally as a combination of Chern classes, as in Witten-Kontsevich intersection theory.

In this paper we pursue the study of this model in two directions. Firstly, in section~\ref{sRelWeilPet} we make precise the relation between our model and Witten-Kontsevich intersection theory. The context is geometry and topology.s
We show that our measure is equivalent to the Weil-Peterson volume form on the moduli space of the sphere with marked points (punctures) $\overline{\mathcal{M}}_{0,n}$. 
This equivalence is a non-trival result. It shows that the analysis leading to point (iv) in \cite{DavidEynard2014} was incomplete, with an incorrect conclusion.
It requires a precise study of the Chern class formulation of \cite{DavidEynard2014} at the boundaries between different domains in moduli space corresponding to different Delaunay triangulations, as well as a study of the relation between our geometrical formulation of the volume form and the so-called $\lambda$-length parametrization of $\overline{\mathcal{M}}_{0,n}$.
This result proves that, at least as far as topological (i.e. global) observables are concerned, our model is in the same universality class than pure topological gravity ($\gamma=\sqrt{8/3}$ Liouville or (3,2) Liouville CFT). 
This was a conjecture of \cite{DavidEynard2014}, which now proven.

Secondly, in section~\ref{sLocIneq} we start to study the specific local properties of the measure of the model, which are related to the specific conformal embedding of a discrete random metric defined by the Delaunay triangulations. Basically nothing is known on these properties (except that it is conformal). The context here is measure and probability theory.
The study of these properties should be crucial to make precise the existence of a local continuum limit for this random Delaunay triangulation model and its relation with the Liouville theory. 

We show in Section~\ref{sMaxPropMes} an interesting property of maximality for the measure: our measure on Delaunay triangulations can be analytically continued to non-Delaunay triangulations based on the same points distributions, but is maximal exactly for Delaunay triangulations  (\textit{i.e.} the weight given to a Delaunay triangulation by this measure is superior to the one given to any other triangulation of the same set of points). This could open the possibility of some convexity properties.
Then in Section~\ref{ssGrowVol} we study local bounds on the measure when one considers the process of adding locally a new vertex, thus going from a triangulations with $N$ vertices to a triangulation with $N+1$ vertices.
We get both local and global lower bounds, and deduce that the partition function $Z_N=V_N/N!$ grows at least like $(\pi^2/8)^N$.
These results are encouraging first steps towards the construction of a continuum limit. Let us stress again that they deal with local properties of the embedding of the triangulation in the plane, not global topological properties of the model.

\section{Reminders}
\subsection{The Delaunay triangulation model}
We recall the notations and definitions of \cite{DavidEynard2014}.
Let $T$ denote an abstract triangulation of the Riemann sphere $\mathcal{S}_2=\mathbb{C}\cup\{\infty\}$.
$\VofT$, $\EofT$ and $\FofT$ denote respectively the sets of vertices $v$, edges $e$ and faces (triangles)  $f$ of $T$.
Let ${\mathcal{T}}_N$ be the set of all such $T$ with $N=|\VofT|$ vertices, hence $|\EofT|=3(N-2)$ and $|\FofT|=2(N-2)$.

An \emph{Euclidean triangulation} $\widetilde T=(T,\boldsymbol{\theta})$ is a triangulation $T$ plus an associated edge angle pattern $\boldsymbol{\theta}=\{\theta(e); e\in \EofT\}$, such that 
\begin{equation}
\label{thetacond}
0\le\theta(e)<\pi\ .
\end{equation}
An Euclidean triangulation is \emph{flat} if for each vertex $v\in \VofT$, the sum of the angles of the adjacent edges satisfy
\begin{equation}
\label{sumthetav}
\sum_{e\to v} \theta(e)=2\pi
\end{equation}

Given a set of $N$ points with complex coordinates $z_v$ in $\mathcal{S}_2$ (with its standard complex structure), the associated Delaunay triangulation in a flat Euclidean triangulation, such that 
the angle $\theta(e)$ is the angle of intersection between the circumcircles of the oriented triangles (faces) adjacent to $e$. 
 See fig. \ref{definition_angles}
The edge angle pattern satisfies in addition the condition that for any closed oriented contour $\mathcal{C}^\star$ on the dual graph $T^\star$ of the triangulation $T$ (the Vorono{\"\i}   diagram), the sum of the angles associated to the edges $e$ dual (orthogonal) to the edges $e^\star$ of $\mathcal{C}^\star$ satisfy
\begin{equation}
\label{sumthetacontour}
\sum_{e\perp \mathcal{C}^\star}\theta(e)\ge 2\pi
\end{equation}
(this condition was not discussed in \cite{DavidEynard2014}). 

\begin{figure}[h]
\begin{center}
\includegraphics[width=3in]{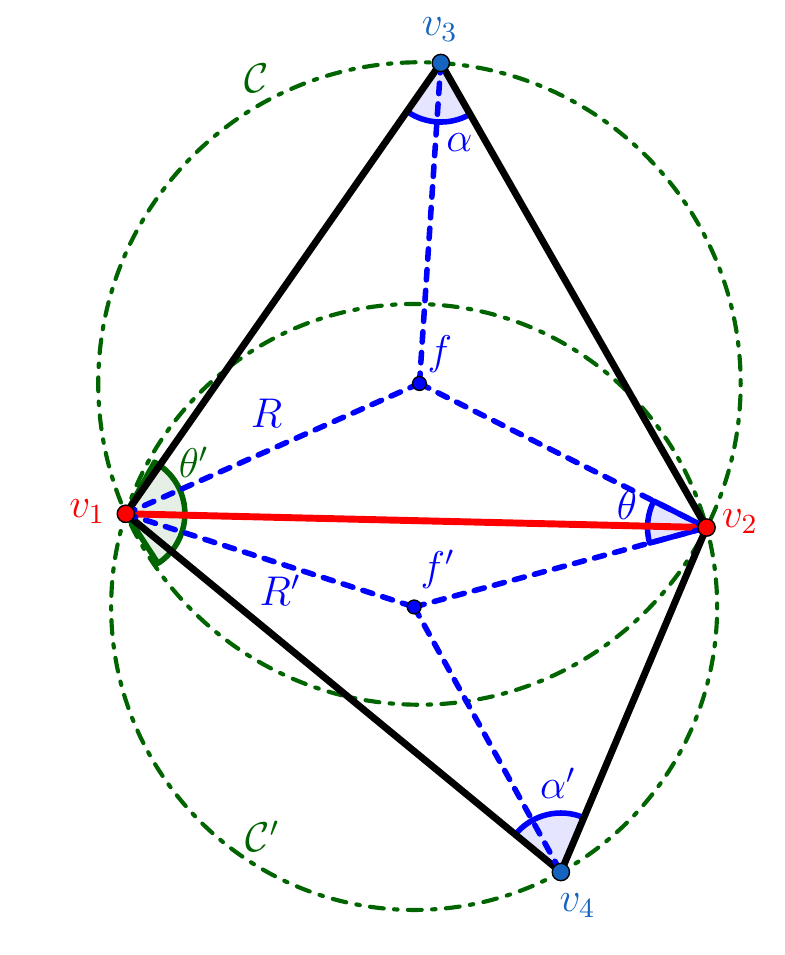}
\caption{The triangles $f$ and $f'$, the circumcircles $\mathcal{C}$ and $\mathcal{C}'$ and angles $\theta$ and $\theta'=\pi-\theta$ associated to an edge $e=(v_1,v_2)$ of a Delaunay triangulation. Here, $R$ and $R'$ are the radii of $\mathcal{C}$ and $\mathcal{C}'$ respectively.}
\label{definition_angles}
\end{center}
\end{figure}

\begin{definition}
\label{defTtilde}
We denote by $\widetilde{\mathcal{T}}_{N}^{f}$ the set of Euclidean triangulations $\tilde{T}=(T,\boldsymbol{\theta})$ with $N$ vertices that satisfy \ref{thetacond}, \ref{sumthetav} and \ref{sumthetacontour}. The subscript ``$f$'' is for ``flat''.
\end{definition}

A theorem by Rivin \cite{Rivin1994} states that this is in fact an angle pattern preserving bijection between $\widetilde{\mathcal{T}}_N^{f}$ and the set of Delaunay triangulations of the complex plane modulo M\"obius transformations, which can be identified with $\boldsymbol{\mathfrak{D}}_N=\mathbb{C}^N/ \mathrm{SL}(2,\mathbb{C})$ (it is angle pattern preserving in the sense that the angles of the Euclidean triangulations are the same as the angles defined by the intersection of the circumcircles of the Delaunay triangulation associated by the bijection).
It is an extension of the famous theorem by Koebe-Andreev-Thurston  \cite{Koebe1936} stating that there is a bijection between simple triangulations and circle packings in complex domains, modulo global conformal transformations. The proof relies on the same kind of convex minimization functional, using hyperbolic 3-geometry,  rather than for the original circle packing case (see \cite{Rivin1994} and \cite{BobenkoSpringborn2004}).

The model of random triangulation considered in \cite{DavidEynard2014} is obtained by taking the discrete uniform measure on triangulations  with $N$ vertices, times the flat Lebesgue measure on the angles.
Since the PSL(2,$\mathbb{C}$) invariance allows to fix 3 points in the triangulations, from now on we work with triangulations and points ensembles with $M=N+3$ points. The measure on $\widetilde{\mathcal{T}}_M^{f}$ is 

\begin{equation}
\label{muT}
\mu(\widetilde{T})=\mu(T,d \boldsymbol{\theta})= {\rm uniform}(T)\ \prod_{e\in \EofT} \hskip -1ex d\theta(e)\quad \hskip -1em\prod_{v\in \VofT}\hskip -1ex \delta\big(\sum_{e\mapsto v} \theta(e)-2\pi\big)
\  \prod_{\mathcal{C}^\star} \Theta\big(\sum_{e\perp \mathcal{C}^\star} \theta(e)-2\pi\big)
\end{equation}
 where, $\Theta(x)=\begin{cases}1 \mathrm{\,\,if\,\,} x\geq0\\
0 \mathrm{\,\, if\,\,} x<0
\end{cases}$ is the Heaviside function.

\subsection{K{\"a}hler form of the measure}
One of the main results of \cite{DavidEynard2014} is the form of the induced measure on the space $\mathfrak{D}_{N+3}$ of Delaunay triangulations on the plane, i.e on the space of distributions of $N+3$ points on the Riemann sphere. The first three points $(z_1,z_2,z_3)$ being fixed by PSL(2,$\mathbb{C}$), the remaining $N$ coordinates are denoted $\mathbf{z}=(z_4,\cdots,z_{N+3}) \in\mathbb{C}^N$, and $T_{\mathbf{z}}=T$ is the associated abstract Delaunay triangulation, uniquely defined if no subset of 4 points are cocyclic. A simple case is when one of the three fixed point is at $\infty$.
\begin{theorem} \cite{DavidEynard2014}
The measure $\mu(T,d\boldsymbol{\theta})\ =\ d\mu(\mathbf{z})$ on $\mathbb{C}^D$ is a K{\"a}hler measure of the local form
\begin{equation}
\label{themeasurez}
d\mu(\mathbf{z})=\prod_{v=4}^{N+3} d^2z_v\ 2^N\ \det\left[{D_{u\bar v}}\right]
\end{equation}
where $D$ is the restriction to the $N$ lines and columns $u,\bar v=4,\,5,\,\cdots\, N+3$ of the K\"ahler metric on $\mathbb{C}^{N+3}$ 
\begin{equation}
\label{Duvbar}
D_{u\bar v}(\{z\})={\frac{\partial}{\partial z_u}}{\frac{\partial}{\partial \bar z_v}} \mathcal{A}_T(\{z\})
\end{equation}
with the prepotential $\mathcal{A}_T$ given by 
\begin{equation}
\label{AsumV}
\mathcal{A}_T= -\sum_{f\in\mathcal{F}(T)} \mathbf{V}(f)
\end{equation}
where the sum runs over the triangles $f$ (the faces) of the Delaunay triangulation $T$ of the Delaunay triangulation $T$ associated to the configuration of points $\{z\}=\{z_v;\,v=1,N+3\}$ in $\mathbb{C}^{N+3}$. For a triangle $f$ with (c.c.w. oriented) vertices $(z_a,z_b,z_c)$,  $\mathbf{V}(f)$ is the hyperbolic volume in the hyperbolic upper half space $\mathbb{H}^3$ of the ideal tetraedron with  vertices $(z_a,z_b,z_c,\infty)$ on its boundary at infinity $\mathbb{C}\cup\{\infty\}$.
\end{theorem}
\noindent NB: This  statement is a bit loose and some care must be taken in treating the point at $\infty$ and the three fixed points. See \cite{DavidEynard2014} for details. 

\subsection{Relation with topological Witten-Kontsevich intersection theory}
A second result of \cite{DavidEynard2014} is that the measure can be written locally (i.e. for a given triangulation $T$) in term of Chern classes $\psi_v$ of U(1) line bundles $L_v\to \mathcal{M}_{0,N+3}$, (attached to the vertices v) where the space of Delaunay triangulations with $N+3$ points is identified with the moduli space $\mathcal{M}_{0,N+3}$ of the (conformal structures of the) sphere with $N+3$ marked points. More precisely it was shown that locally
\begin{equation}
\label{mupsiN}
\mu(T,\{d\theta\})= {\frac{1}{N!} \, 2^{2N+1}} \left(\sum_v (2\pi)^2 \psi_v\right)^N
\end{equation}
with 
\begin{equation}
\label{psivdef}
\psi_v=c_1(L_v)= {\frac{1}{(2\pi)^2}} \sum_{e'<e\to v}d\theta(e) \wedge d\theta(e')
\end{equation}
where the $e$'s denote the edges adjacent to the vertex $v$, labelled in c.c.w. order.  With this convention, the notation $e'<e\to v$ means that the sum runs on the pairs of edges $e,\,e'$ adjacent to the vertex $v$ and such that their labels satisfy $e'<e$.
$\psi_v$ was defined explicitly as the curvature $d u_v$ of the global U(1) connection
\begin{equation}
\label{uvexpl}
u_v= {\frac{1}{(2\pi)^2}}\sum_{f\to v} \theta(f_+)\, d\gamma_v(f)
\end{equation}
In \ref{uvexpl} the sum runs over the faces $f$ adjacent to the vertex v. 
$\gamma_v(f)$ is the angle between a reference half-line $\gamma_v$ with endpoint $v$ and the half line starting from $v$ and passing through the center of the (circumcircle of the) face $f$. $f_+$ is the leftmost edge of $f$ adjacent to $v$ (see figure \ref{f-uvexpl}).
\begin{figure}[h]
\begin{center}
\includegraphics[width=3in]{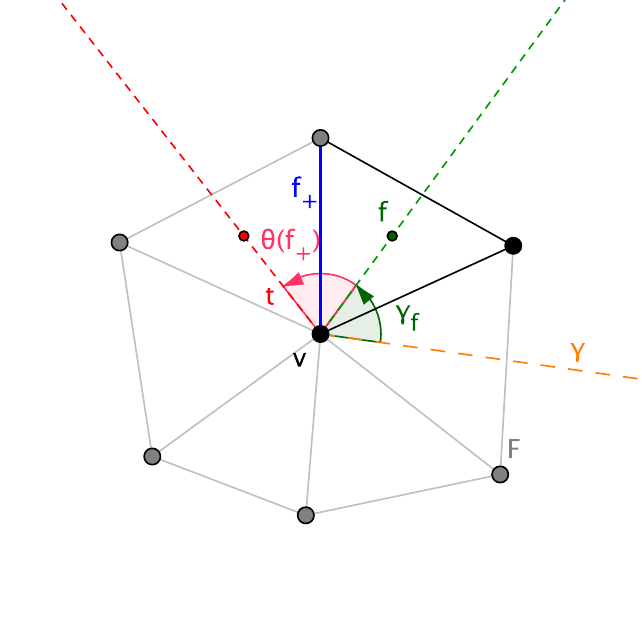}
\caption{Construction of the connection $u_v$}
\label{f-uvexpl}
\end{center}
\end{figure}
It was stated in \cite{DavidEynard2014} that the measure \ref{themeasurez} is therefore the measure of topological gravity studied in \cite{Kontsevich1992}.  As we shall see in the next section, this is a local result, but it cannot be extended globally.

\section{Relation with Weil-Petersson Metric}
\label{sRelWeilPet}
\subsection{Discontinuities of the Chern Classes}
\label{ssDiscont}
There is a subtle point in establishing the link between the angle measure on Delaunay triangulations (let us denote this measure $\mu_\mathscr{Del.}$)  and the measure of the topological Witten-Kontsevich intersection theory, that we shall denote $\mu_{\mathrm{top.}}$
In its general definition through \ref{mupsiN}
\begin{equation}
\label{ }
\mu_{\mathrm{top.}}\propto \psi^N\ ,\quad \psi=\sum_v \psi_v= du
\ ,\quad
u=\sum_v u_v
\end{equation}
the curvature 2-form $\psi$ and the 1-form $u$ (the global U(1) connection) depend implicitely on a choice of triangulation $T$ of the marked sphere, which is supposed to be kept fixed, but the final measure $\mu$ and its integral over the moduli space does not depend on the choice of triangulation).

In our formulation, the moduli space $\mathcal{M}_{0,N+3}$ is the closure of the union of disjoint domains $\mathcal{M}^{(T)}$ where the triangulation $T$ is combinatorially a given Delaunay triangulation. Two domains  $\mathcal{M}^{(T)}$ and $\mathcal{M}^{(T')}$ meet along a face (of codimension 1) where the four vertices of two  faces sharing an edge are cocyclic, so that one passes from $T$ to $T'$ by a flip, as depicted on Fig.~\ref{flipangles}. 
The relation $\mu_\mathscr{Del.}=\mu_\mathrm{top.}$ will be valid if the form $u$ is continuous along a flip. If this is not the case, there might be some additional boundary terms in $du$.

\begin{figure}[h]
\begin{center}
\includegraphics[width=3in, trim=.45in .5in .5in .5in]{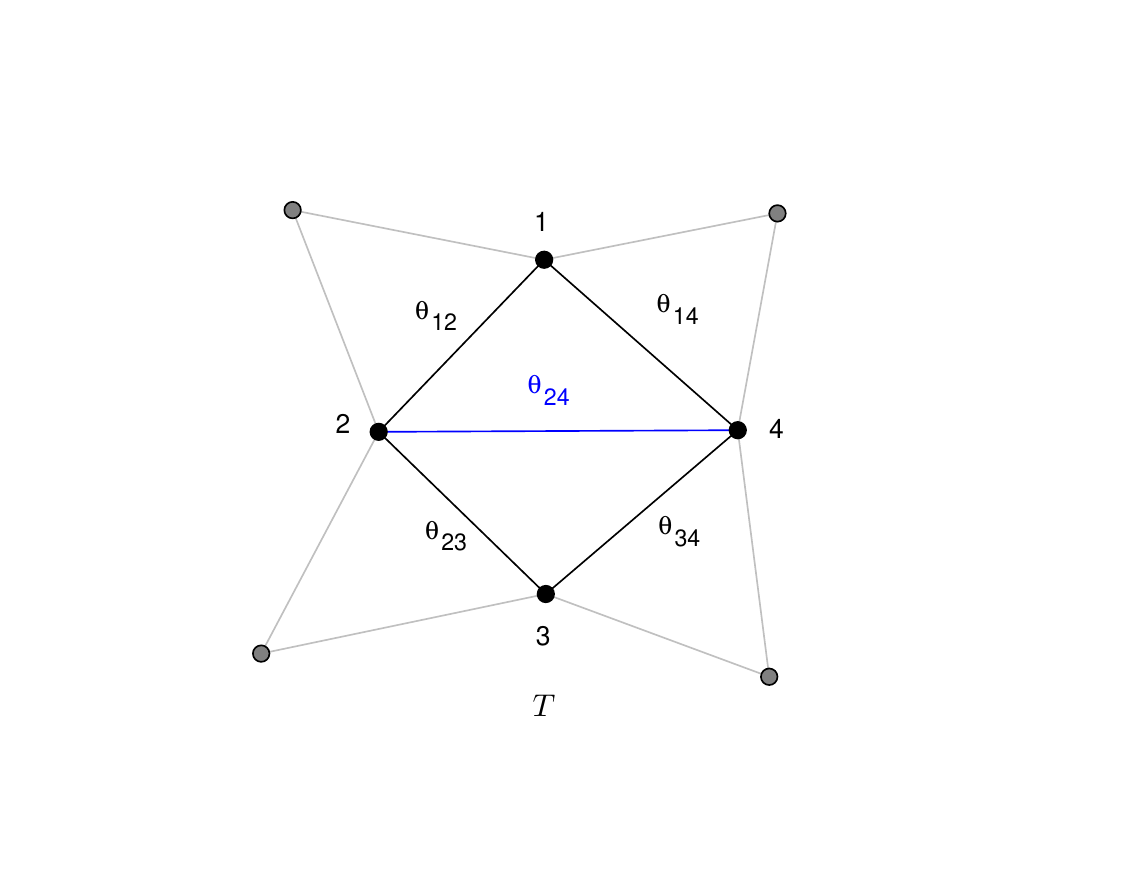}\includegraphics[width=3in, trim=.45in .5in .5in .5in]{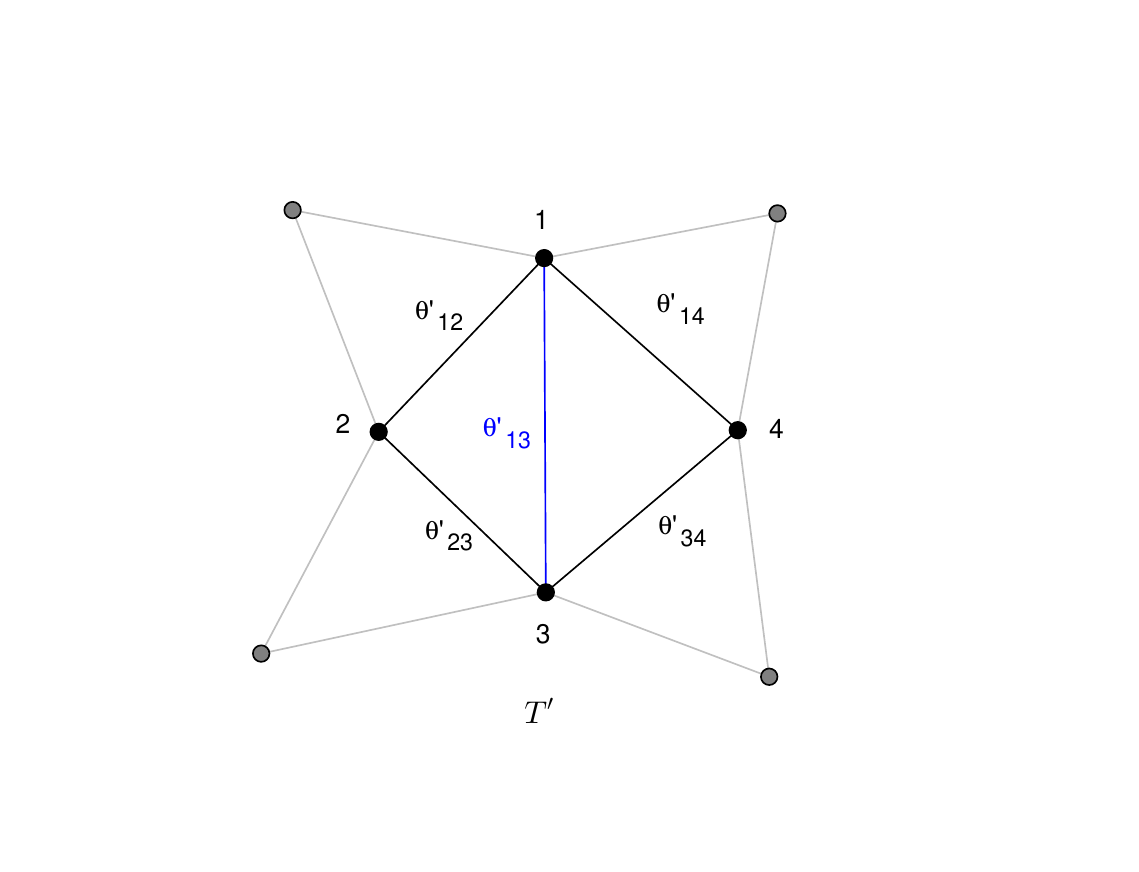}
\caption{ }
\label{flipangles}
\end{center}
\end{figure}
Let us therefore compare the 2-form $u$ for a triangulation $T$ and the corresponding 2-form $u'$ for the triangulation $T'$ obtained from $T$ by the flip $(2,4)\to(1,3)$ depicted on Fig.~\ref{flipangles}. The angles $\theta$ of the edges of $T$ and $\theta'$ of the edges of $T'$ are a priori different for the five edges depicted here (when the points 1, 2, 3 an 4 are not cocyclic) but only six among the ten angles are independent, since they satisfy the relation at vertex 1
\begin{equation}
\label{ }
\theta_{12}+\theta_{14}=\theta'_{12}+\theta'_{13}+\theta'_{14}
\end{equation}
and the three similar relations for vertex 2, 3 and 4. These relations imply for instance that $\theta_{24}+\theta'_{13}=0$.
From the definition \ref{uvexpl} of the 1-forms $u$ and $u'$ one computes easily $u-u'$ (which depends a priori on the choice of section angles $\gamma_1,\,\ldots\,\gamma_4$). 
However we are interested at the difference at the flip between the Delaunay triangulations $T$ and $T'$, i.e when the 4 points are cocyclic. Then $\theta_{12}=\theta'_{12}$, $\theta_{14}=\theta'_{14}$, $\theta_{23}=\theta'_{23}$, $\theta_{34}=\theta'_{34}$ and $\theta_{24}=\theta_{13}=0$ and we get\begin{equation}
\label{u-uprime}
\left.u-u'\right|_\mathrm{flip}=(\theta_{14}+\theta_{23}-\theta_{12}-\theta_{34})(d\theta_{12}-d\theta'_{12}) + (\theta_{14}+\theta_{23}) d\theta_{24}
\end{equation}
This is clearly non zero.
 Despite of the apparent dihedral symmetry breaking of formula \ref{u-uprime}, it is actually symmetric, and using the relations holding at the vertices, it is equivalent to:
\bea 
    \label{}
    \left.u-u'\right|_\mathrm{flip}&=&\frac{1}{2}\left[(\theta_{12}+\theta_{34})(d\theta'_{12}-d\theta_{12}+d\theta'_{34}-d\theta_{34})\right.\cr
    && \left. + (\theta_{14}+\theta_{23})(d\theta'_{14}-d\theta_{14}+d\theta'_{23}-d\theta_{23})\right]
\eea
The 1-form of topological Witten-Kontsevich intersection theory $u_{\scriptscriptstyle\mathrm{top.}}$ is defined globally as a sum over the triangulations T  as
\begin{equation}
\label{ }
u_{\scriptscriptstyle\mathrm{top.}}=\sum_T \chi_{\scriptscriptstyle (T)} u_{\scriptscriptstyle (T)}
\end{equation} 
where $\chi_{\scriptscriptstyle (T)}$ is the indicator function (hence a 0-form) of the domain 
${\mathcal{M}}^{\scriptscriptstyle (T)}$
and $u_{\scriptscriptstyle (T)}$ the 1-form for the triangulation $T$. 
The 2-form (Chern class) of the topological gravity theory  of \cite{Witten:1990} and \cite{Kontsevich1992}
is therefore
\begin{equation}
\label{psitopexpl}
\psi_{\scriptscriptstyle\mathrm{top.}}= d u_{\scriptscriptstyle\mathrm{top.}} 
=\sum_T\  \chi_{\scriptscriptstyle (T)}\, du_{\scriptscriptstyle (T)}+ 
d\chi_{\scriptscriptstyle (T)}\wedge u_{\scriptscriptstyle (T)}
\end{equation}
In \cite{DavidEynard2014} it was shown that the Delaunay measure can be written locally (inside each $\mathcal{M}^T$) as the volume form of the Delaunay 2-form
\begin{equation}
\label{psiDelauGlob}
 \psi_{\scriptscriptstyle\mathscr{Del.}}= \sum_T \chi_{\scriptscriptstyle (T)} du_{\scriptscriptstyle (T)}
\end{equation}
and that this measure is continuous at the boundary between two adjacents domains $\mathcal{M}^T$ and $\mathcal{M}^{T'}$, so that the definition \ref{psiDelauGlob} is global.

The calculation leading to \ref{u-uprime} shows that the 1-form $u$ is generically not continuous at the boundary between domains, so that the second term in \ref{psitopexpl} is generically non-zero, and localized at the boundaries between domains.
Therefore the Delaunay 2-form $\psi_{\scriptscriptstyle\mathscr{Del.}}$ is different from the topological 2-form $\psi_{\scriptscriptstyle\mathrm{top.}}$ (the Chern class) and the associated measures (top forms) are a priori different
\begin{equation}
\label{ }
\mu_{\scriptscriptstyle\mathrm{top.}}\,=\, c_N\ (\psi_{\scriptscriptstyle\mathrm{top.}})^N
\ \neq\ 
\mu_{\scriptscriptstyle\mathscr{Del.}}\,=\,c_N\ (\psi_{\scriptscriptstyle\mathscr{Del.}})^N
\end{equation}
the difference being localized at the boundaries of the domains $\mathcal{\overline M}^{\scriptscriptstyle (T)}_{0,N+3}$.

\subsection{The angle measure and the Weil-Petersson metric}
\label{ssWPmetric}
\subsubsection{The Delaunay K\"ahler form}
For a given Delaunay triangulation $\tilde T$, the Delaunay K\"ahler metric form is
\begin{equation}
\label{ }
G_\mathscr{D.}(\{z\})= dz_u d\bar z_v \, D_{u\bar v}(\{z\})
\end{equation}
with $D_{u\bar v}(\{z\})$ given by \ref{Duvbar}.
The associated Delaunay K\"ahler 2-form is
\begin{equation}
\label{ }
\Omega_\mathscr{D.}(\{z\})= {\frac{1}{2 \mathrm{i}}}dz_u \wedge d\bar z_v \, D_{u\bar v}(\{z\})
\end{equation}
$G_\mathscr{D.}$ and $\Omega_\mathscr{D.}$ are continuous across flips.
From \cite{DavidEynard2014} the matrix $D$ is
\begin{equation}
\label{ }
D=\frac{1}{4\mathrm{i}} A E A^\dagger
\end{equation}
with $A$ the $(N+3)\times 3(N+1)$ vertex-edge matrix 
\begin{equation}
\label{AueDef}
A_{ue}=\begin{cases}
      \frac{1}{z_u-z_{u'}}& \text{if $u$ is an end point of the edge $e=(u,u')$ of $T$}, \\
      \ \ 0& \text{otherwise}.
\end{cases}
\end{equation}
and $E$ the $3(N+1)\times 3(N+1)$ antisymmetric matrix
\begin{equation}
\label{Eee'Def}
E_{e e'}=\begin{cases}
   +1   & \text{ if $e$ and $e'$ consecutives edges of a face $f$, in c.w. order}, \\
   -1   & \text{ if $e$ and $e'$ consecutives edges of a face $f$, in c.c.w. order}, \\
   \hphantom{+} 0  & \text{otherwise}.
\end{cases}
\end{equation}
Then, the 2-form $\Omega_\mathscr{D.}$ takes a simple form, as a sum over faces (triangles) $f$ of $T$. 
Let us denote $(f_1,f_2,f_3)$ the vertices of a triangle $f$, in c.c.w. order (this is defined up to a cyclic permutation of the 3 vertices).
\begin{equation}
\label{OmegaSumf}
\Omega_\mathscr{D.}(\{z\})=\sum_{\mathrm{faces}\,f} \omega_{\scriptscriptstyle\mathscr{D.}}(z_{f_1},z_{f_2},z_{f_3})
\end{equation}with, for a face $f$ with vertices labelled $(1,2,3)$ (for simplicity), and denoting $z_{ij}=z_j-z_i$
\begin{equation}
\label{omegadef}
\omega_{\scriptscriptstyle\mathscr{D.}}(z_1,z_2,z_3)=\frac{1}{ 8}
\left(\begin{array}{c}
       \hphantom{+\,}d\,\log(z_{23})\wedge d\,\log(\bar z_{31})+ d\,\log(\bar z_{23})\wedge d\,\log(z_{31})\\
       +\, d\,\log(z_{31})\wedge d\,\log(\bar z_{12})+ d\,\log(\bar z_{31})\wedge d\,\log(z_{12})\\
       +\, d\,\log(z_{12})\wedge d\,\log(\bar z_{23})+ d\,\log(\bar z_{12})\wedge d\,\log(z_{23})
\end{array}\right)
\end{equation}
Reexpressed in term of the log of the modulus and of the argument of the $z_{ij}$'s 
\begin{equation}
\label{ }
\lambda_{ij}=\log(|z_j-z_i)|\ ,\quad \vartheta_{ij}=\arg(z_j-z_i)
\end{equation}
we obtain
\begin{equation}
\label{ }
\omega_{\scriptscriptstyle\mathscr{D.}}=\omega_{\mathrm{length}}+\omega_{\mathrm{angle}}
\end{equation}
with the length contribution
\begin{equation}
\label{ }
\omega_{\mathrm{length}}=\frac{1}{4}(d\,\lambda_{12}\wedge d\,\lambda_{23}+d\,\lambda_{23}\wedge d\,\lambda_{31}+d\,\lambda_{31}\wedge d\,\lambda_{12})
\end{equation}
and the angle contribution
\begin{equation}
\label{ }
\omega_{\mathrm{angle}}=\frac{1}{ 4}(d\,\vartheta_{12}\wedge d\,\vartheta_{23}+d\,\vartheta_{23}\wedge d\,\vartheta_{31}+d\,\vartheta_{31}\wedge d\,\vartheta_{12})
\end{equation}
Reexpressed in terms of the angles $\alpha_1$, $\alpha_2$ and $\alpha_3$ of the triangle $(1,2,3)$ (using $\alpha_1=\vartheta_{13}-\vartheta_{12}$, etc.), and using $\alpha_1+\alpha_2+\alpha_3=\pi$, one has
\begin{equation}
\label{ }
\omega_{\mathrm{angle}}=\frac{1}{4}(d\,\alpha_{1}\wedge d\,\alpha_{2})
=\frac{1}{4} (d\,\alpha_{2}\wedge d\,\alpha_{3})=\frac{1}{4} (d\,\alpha_{3}\wedge d\,\alpha_{1})
\end{equation} 
Using the triangle relation
\begin{equation}
\label{ }
{\sin(\alpha_1)\over \exp(\lambda_{23})}={\sin(\alpha_2)\over \exp(\lambda_{31})}={\sin(\alpha_3)\over \exp(\lambda_{12})}
\end{equation}
one gets
\begin{align}
\label{ }
d\alpha_1 \cot\alpha_1-d\lambda_{23}=&d\alpha_2 \cot\alpha_2-d\lambda_{31}=d\alpha_3 \cot\alpha_3-d\lambda_{12}\nonumber\\
&=(d\alpha_1+d\alpha_2){\cot\alpha_1\,\cot\alpha_2-1\over\cot\alpha_1+\cot\alpha_2}-d\lambda_{12}
\end{align}
which gives
\begin{align}
\label{}
  d\alpha_1  & = {\csc^2\alpha_2\over \cot\alpha_1+\cot\alpha_2}(d\lambda_{23}-d\lambda_{12}) +
  {\cot\alpha_1\,\cot\alpha_2-1\over\cot\alpha_1+\cot\alpha_2}(d\lambda_{31}-d\lambda_{12}) \\
  d\alpha_2  &  = {\cot\alpha_1\,\cot\alpha_2-1\over\cot\alpha_1+\cot\alpha_2}
  (d\lambda_{23}-d\lambda_{12}) + {\csc^2\alpha_1\over \cot\alpha_1+\cot\alpha_2}(d\lambda_{31}-d\lambda_{12})
\end{align}
which implies
\begin{equation}
\label{ }
d\,\alpha_{1}\wedge d\,\alpha_{2}=d\,\lambda_{12}\wedge d\,\lambda_{23}+d\,\lambda_{23}\wedge d\,\lambda_{31}+d\,\lambda_{31}\wedge d\,\lambda_{12}
\end{equation}
Hence $\omega_{\mathrm{angle}}=\omega_{\mathrm{length}}$.
Therefore one has 
\begin{equation}
\label{wDelFinal}
\omega_{\scriptscriptstyle\mathscr{D.}}={1\over 2}(d\,\lambda_{12}\wedge d\,\lambda_{23}+d\,\lambda_{23}\wedge d\,\lambda_{31}+d\,\lambda_{31}\wedge d\,\lambda_{12})
\end{equation}
\begin{figure}[h]
\begin{center}
\includegraphics[width=2.5in]{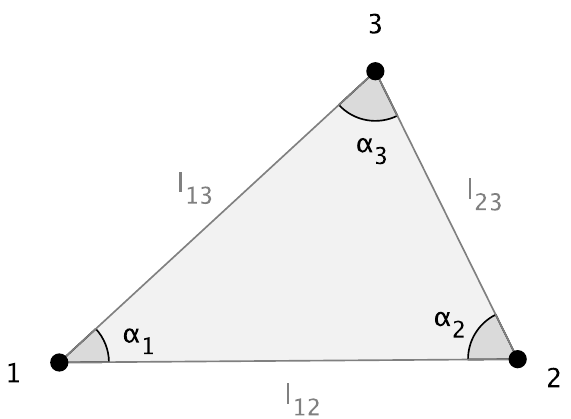}\qquad
\raisebox{4.ex}{\includegraphics[width=2.5in]{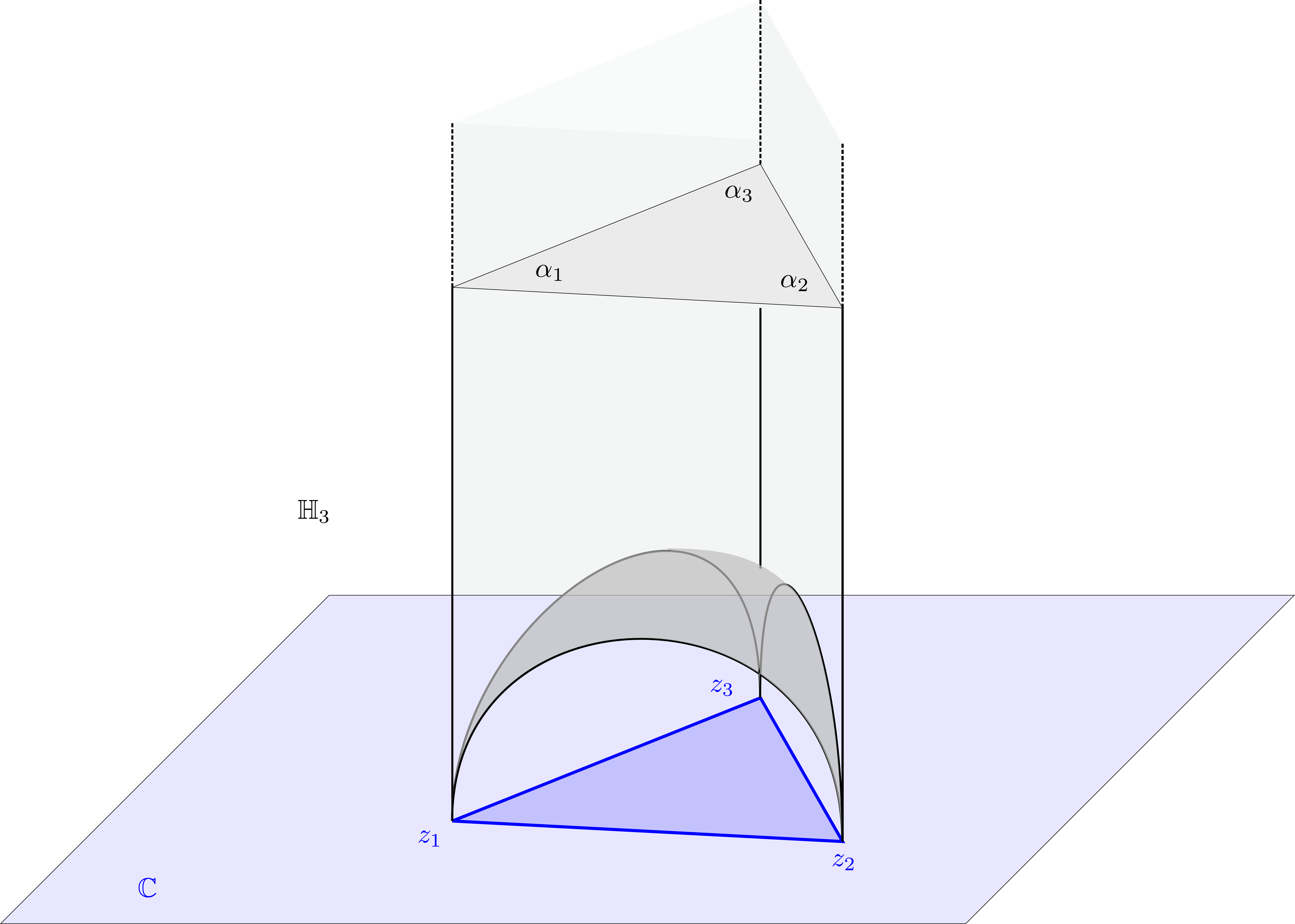}}
\caption{A triangle $f=(1,2,3)$ (left) and the associated ideal spherical triangle $\mathcal{S}_{123}$ in $\mathbb{H}_3$ (right).}
\label{ }
\end{center}
\end{figure}

\subsubsection{Delaunay triangulations and moduli space}

We can now compare this structure with the Weil-Petersson K\"ahler structure on the decorated moduli space  $\mathcal{\tilde M}_{0,N+3}$ of the punctured sphere, decorated by horocycles. We refer to \cite{Penner1987}  and to \cite{Penner2006}, \cite{Thurston2012}, \cite{HoTT_vol_1} (among many references) for a general introduction to the subject.

To any Delaunay triangulation $\tilde T$ with $N+3$ points on the complex plane, we can associate an explicit surface $\mathcal{S}$ with constant negative curvature  and $N+3$ punctures as follows.
Let $\mathbb{H}_3=\mathbb{C}\times \mathbb{R}^{^*}_+$ be the upper half-space above $\mathbb{C}$, with coordinates $(z,h)$ embodied with the Poincar\'e metric $ds^2= (dzd\bar z+ dh^2)/h^2$. It makes $\mathbb{H}_3$ the 3-dimensional hyperbolic space, with $\mathbb{C}\cup\{\infty\}$ its asymptotic boundary at infinity.

Consider a triangle $f_{123}$ with vertices $(1,2,3)$ (in c.c.w. order) with complex coordinates $(z_1,z_2,z_3)$ in $\mathbb{C}$. 
Let $\mathcal{B}_{123}$ be the hemisphere in $\mathbb{H}_3$ whose center is the center of the circumcircle of $f_{123}$ (in $\mathbb{C}$), and which contains the points $(1,2,3)$.
$\mathcal{B}_{123}$, embodied with the restriction of the Poincar\'e metric $ds^2$ of $\mathbb{H}_3$, is isometric to the 2 dimensional hyperbolic disk $\mathbb{H}_2$.
Let $\mathcal{L}_{12}$ be the intersection of $\mathcal{B}_{123}$ with the half plane orthogonal to $\mathbb{C}$ which contains the points $1$ and $2$, this is a semicircle orthogonal to $\mathbb{C}$. With a similar definition for $(23)$ and $(31)$, the semicircles $\mathcal{L}_{12}$, $\mathcal{L}_{23}$ and $\mathcal{L}_{31}$ delimit a 
spherical triangle $\mathcal{S}_{123}$ on the hemisphere in $\mathbb{H}_3$. 
The semicircles $\mathcal{L}_{12}$, $\mathcal{L}_{23}$ and $\mathcal{L}_{31}$ are geodesics in $\mathbb{H}_3$, hence in $\mathcal{B}_{123}$, so that $\mathcal{S}_{123}$ is an ideal triangle in $\mathbb{H}_2$.
$\mathcal{S}_{123}$ is nothing but the face $(123)$ of the ideal tetraedra $(z_1,z_2,z_3,\infty)$ in $\mathbb{H}_3$ whose volume $\mathbf{V}(f)$ appears in \ref{AsumV}.

Now consider a Delaunay triangulation $\tilde T$ in the plane, with $N+3$ points, and with one point at infinity for simplicity.
The union of the ideal spherical triangles $\mathcal{S}_f$ associated to the faces $f$ of $\tilde T$ form surface $\mathcal{S}$ in $\mathbb{H}_3$
\begin{equation}
\label{ }
\mathcal{S}=\bigcup_{f\in\mathcal{F}(T)}\mathcal{S}_f
\end{equation}
See fig.~\ref{surface-triangles}.
The surface $\mathcal{S}$ embodied with the  restriction of the Poincar\'e metric of $\mathbb{H}_3$, is a constant negative curvature surface. Indeed since the triangles $\mathcal{S}_f$ are glued along geodesics, no curvature is localized along the edges of these triangles. It is easy to see that the endpoints $z_i$ of the triangulations are puncture curvature singularities of $\mathcal{S}$, i.e. points where the metric can be written (in local conformal coordinates with the puncture at the origin) 
\begin{equation}
\label{ }
ds^2= {dwd\bar w\over |w|^2 |\log(1/|w|)|^2}
\end{equation}
\begin{figure}[h]
\begin{center}
\includegraphics[width=12.cm]{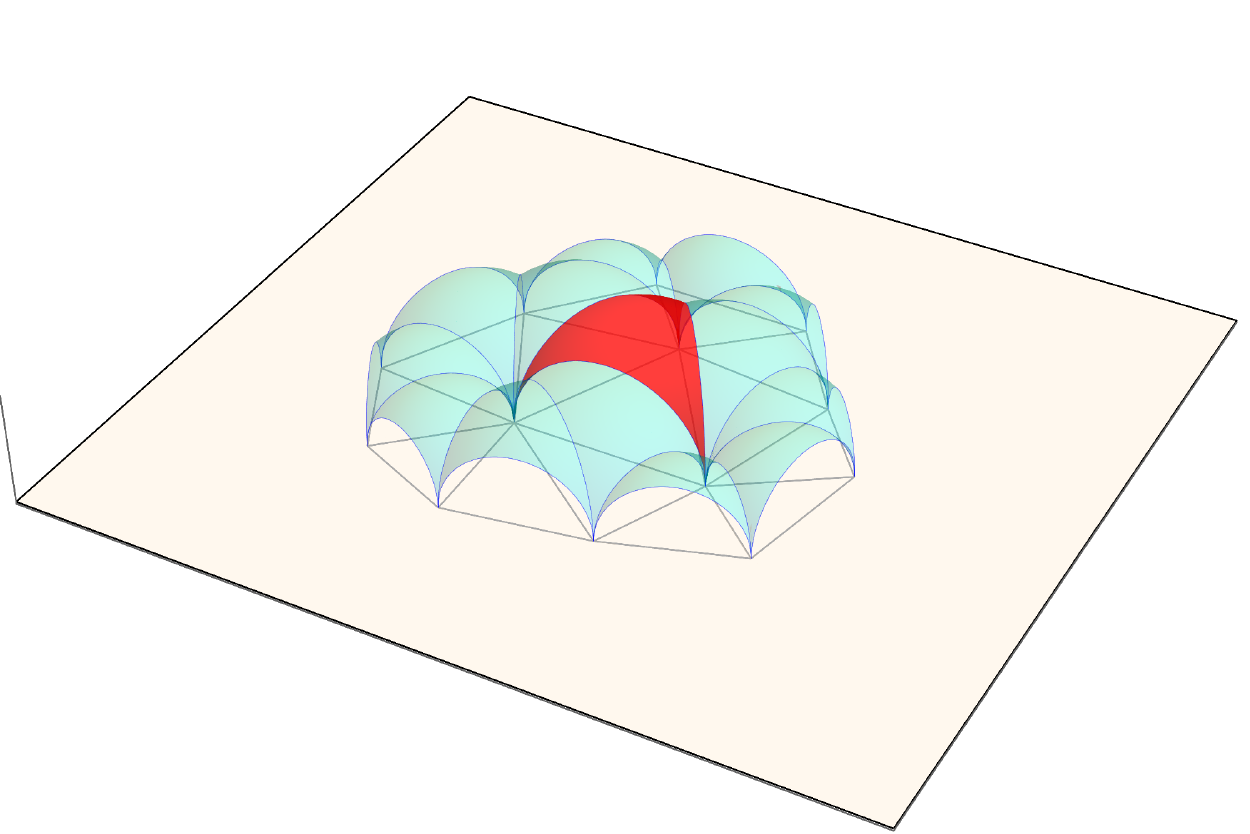}
\caption{A triangulation and the associated punctured surface}
\label{surface-triangles}
\end{center}
\end{figure}

Through the orthogonal projection from $\mathcal{S}$ to the plane $\mathbb{C}$, the metric in each $\mathcal{S}_f$ become the standard Beltrami-Cayley-Klein hyperbolic metric in the triangle $f$. 
We recall that it is defined in the unit disk $\mathbb{D}_2=\{z; |z|<1\}$ in radial coordinates as
\begin{equation}
\label{ }
ds^2_{\mathscr{B.}}={dr^2 + r^2 d\theta^2 \over (1-r^2)}+{(r\,dr)^2\over (1-r^2)^2}
\end{equation}
that it is not conformal, and is such that geodesics are straight lines in the disk.
Thus, each Delaunay triangulation -- modulo PSL(2,$\mathbb{C}$) tranformations -- gives explicitly,  the constant curvature surface representative of a point in $\mathcal{M}_{0,N+3}$.

\subsubsection{$\lambda$-lengths and horospheres}
Following \cite{Penner1987}, decorated surfaces are obtained by supplementing each puncture  $v$ by a horocycle $\mathfrak{h}_v$, i.e. a closed curve orthogonal to the geodesics emanating from $v$ (in the constant curvature metric). Horocycles are uniquely characterized by their length $\ell_v$. The moduli space of decorated punctured surfaces is simply
\begin{equation}
\label{ }
\mathcal{\tilde M}_{g,n}=\mathcal{M}_{g,n}\otimes \mathbb{R}_+^{^{\scriptstyle{\otimes_v}}}
\end{equation}
A geodesic triangulation $\mathfrak{T}$ of the abstract surface $\mathcal{S}$ is a triangulation such that the edges are (infinite length) geodesics joining the punctures, and the triangles are c.c.w. oriented (and non-overlapping). 
For a decorated surface $\mathcal{\tilde S}$, for any geodesics $\mathfrak{e}$
joining two punctures $u$ and $v$ (generically one may have $u=v$), its  $\lambda$-length ${\Lambda}_\mathfrak{e}(u,v)$ is defined from the (finite, algebraically defined) geodesic distance $d_\mathfrak{e}(u',v')$ along $\mathfrak{e}$ between the intersections $u'$ and $v'$ of $e$ with the horocycles  $\mathfrak{h}_u$ and $\mathfrak{h}_v$ by
\begin{equation}
\label{ }
\Lambda_\mathfrak{e}(u,v)=\exp(d_\mathfrak{e}(u',v')/2)
\end{equation}
For a given triangulation $\mathfrak{T}$ (of a genus $g$ surface with $n$ punctures), it is known that the set of the independent $\lambda$-lengths $\Lambda_e\in\mathbb{R}_+$ for the $6g+3n-6$ edges of $\mathfrak{T}$ provide a complete set of coordinates for the decorated Teichm\"uller space $\mathcal{\tilde T}_{g,n}$ (the universal cover of $\mathcal{\tilde M}_{g,n}$). This parametrization is independent  of the choice of triangulation, thanks to the Ptolemy's relations between lambda-lengths 
when one passes from a triangulation $\mathfrak{T}$ to another one $\mathfrak{T'}$ through a flip similar to the ones of fig.~\ref{flipangles}, namely 
\begin{equation}
\label{Ptolemy}
\Lambda_{13}\Lambda_{24}=\Lambda_{12}\Lambda_{34}+\Lambda_{14}\Lambda_{23}
\end{equation}

In this parametrization, the Weil-Petersson 2-form on $\mathcal{M}_{g,n}$ (through its projection from $\mathcal{\tilde T}_{g,n}$) can be written simply as a sum over the $2(2g+n+2)$ oriented faces (triangles) $\mathfrak{f}$ of $\mathfrak{T}$, as
\begin{equation}
\label{WPLambda}
\Omega_{\scriptscriptstyle{\mathscr{W\!\!P}}} = \sum_{\mathrm{faces}\,\mathfrak{f}} d\log(\Lambda_{12})\wedge d\log(\Lambda_{23})+d\log(\Lambda_{23})\wedge d\log(\Lambda_{31})+d\log(\Lambda_{31})\wedge d\log(\Lambda_{12})
\end{equation}
where 
(1,2,3) denote the vertices (punctures) $v_1(f)$, $v_2(f)$ and $v_3(f)$ (here in c.c.w. order) of the geodesic triangle $\mathfrak{f}$ of $\mathfrak{T}$, and the $\Lambda_{ij}$ denote the $\lambda$-length of the edges of the triangle.

\begin{figure}[h]
\begin{center}
\includegraphics[width=12.cm]{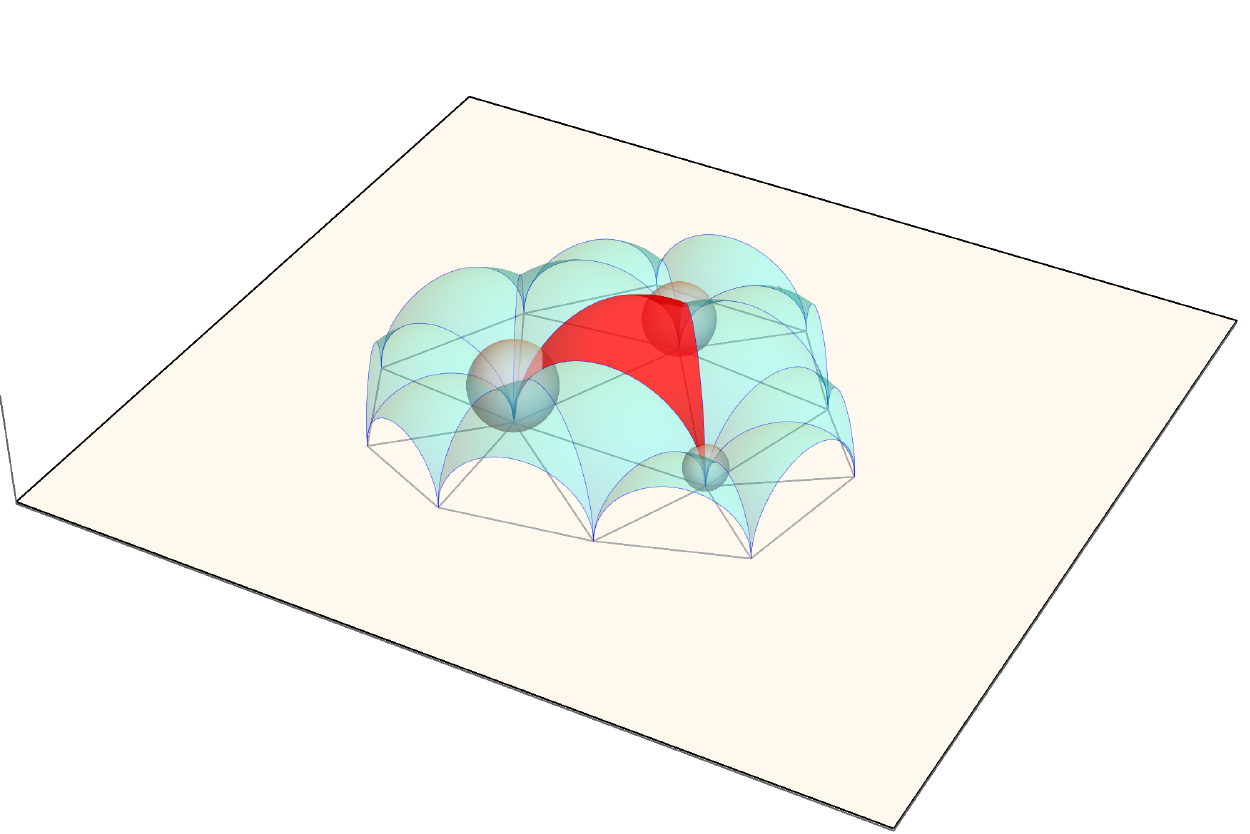}
\caption{The punctured surface decorated with horospheres. Although this 3d representation looks non smooth at the edges, the intrinsic metric of the surface is indeed a smooth constant curvature metric.}
\label{surface+horospheres}
\end{center}
\end{figure}

To compare $\Omega_{\scriptscriptstyle{\mathscr{W\!\!P}}}$ to $\Omega_\mathscr{D.}$, one simply has to look at horocycles and $\lambda$-lengths in Delaunay triangulations. We have an explicit representation of a point in $\mathcal{M}_{0,N+3}$ as the constant curvature surface $\mathcal{S}$ in $\mathbb{H}_3$ constructed above the Delaunay triangulation $T$ for the set of points $\mathbf{z}=\{z_i\}_{i=1,N+3}$ in the complex plane.
Horocycles are easily constructed by decorating each point (puncture) $z_i$ by a horosphere $\mathcal{H}_i$, i.e. an Euclidean sphere in $\mathbb{R}^3$, tangent to the complex place $\mathbb{C}$ at the point $z_i$, and lying above $z_i$ (i.e. in $\mathbb{H}_3$). 
The intersection (in $\mathbb{H}_3$) of the horosphere $\mathcal{H}_i$ with the union of the ideal spherical triangles $\mathcal{S}_f$ for the faces $f$ which share the vertex $i$ defines the horocircle $\mathfrak{h}_i$ associated to the puncture $i$ of $\mathcal{S}$.
See fig.~\ref{surface+horospheres}.

Let $R_i$ denote the Euclidean radius of the horosphere $\mathcal{H}_i$ above vertex $i$. The $\lambda$-length for the edge joining two vertices ($i,j$) of the triangulation is easily calculated (applying for instance the formula in the Poincar\'e half-plane in 2 dimensions) and is 
\begin{equation}
\label{ }
\Lambda(i,j)={|z_i-z_j|\over \sqrt{4\, R_i\, R_j}}
\end{equation}
where $|z_i-z_j|$ is the Euclidean distance between the points $i$ and $j$ in the plane $\mathbb{C}$.
See fig.~\ref{horocycle2D}.
\begin{figure}[h]
\begin{center}
\includegraphics[width=3in]{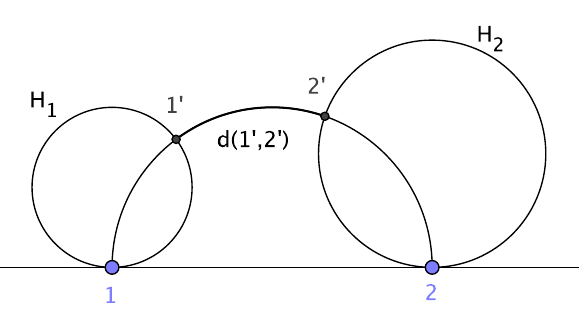}
\caption{The geodesics between the horospheres $H_1$ and $H_2$ at points $1$ and $2$.}
\label{horocycle2D}
\end{center}
\end{figure}

\subsubsection{Identity of the K\"ahler structures}
Incorporating this into \ref{WPLambda}, the  Weil-Petersson 2-form  \ref{WPLambda} the takes a  form similar to that of \ref{wDelFinal} 
\begin{equation}
\label{ }
\Omega_{\scriptscriptstyle{\mathscr{W\!\!P}}} = \sum_{\mathfrak{f}} 
{d\,|z_1-z_2]\over |z_1-z_2|}\wedge {d\,|z_2-z_3]\over |z_2-z_3|}
+
{d\,|z_2-z_3]\over |z_2-z_3|}\wedge {d\,|z_3-z_1]\over |z_3-z_1|}
+
{d\,|z_3-z_1]\over |z_3-z_1|}\wedge {d\,|z_1-z_2]\over |z_1-z_2|}
\end{equation}
On one side \ref{WPLambda} refers to a given geodesic triangulation, but the resulting 2-form $\Omega_{\scriptscriptstyle{\mathscr{W\!\!P}}}$ is known to be independent of the triangulation through the Ptolemy's relation. 
On the other side \ref{wDelFinal} refers to a given Delaunay triangulation, but we know from \cite{DavidEynard2014} that since the matrix $D$ is continuous, $\omega_{\scriptscriptstyle\mathscr{D.}}$ is continuous when one passes from a Delaunay triangulation to another one through flips when four points are cocyclic. Hence we have the global identity
\begin{equation}
\label{Del=WP}
\Omega_{\scriptscriptstyle\mathscr{D.}}={1\over 2} \Omega_{\scriptscriptstyle{\mathscr{W\!\!P}}}
\end{equation} 

We thus have shown that the K\"ahler structure constructed out of the circle pattern associated to random Delaunay triangulations of the sphere is equivalent to the Weil-Petersson K\"ahler structure on the moduli space of the sphere with marked points.
This implies of course that the volume measures are identical (up to a factor $2^{-N}$) and in particular that the total volume of the space of planer Delaunay triangulations with $N+3$ points and the Weil-Petersson volume of the moduli space of the sphere with $N+3$ punctures are proportional.

\subsection{Discussion of the results of \ref{ssWPmetric}}

\subsubsection{Relation between Random Delaunay triangulations and Moduli spaces}

\paragraph{Delaunay and Weil-Petersson:}
Let us firstly mention that a possible relation between the Delaunay triangulation mesure and the Weil-Petersson volume form was already suggested to the authors of\cite{DavidEynard2014} by T. Budd, on the basis of his analytical and numerical calculations of the total volume of the set of Delaunay triangulations for small number of points \cite{Budd2014}, and comparison with the volume of $\overline{\mathcal{M}}_{0,n}$ (the Weil-Petersson volume of moduli space of the punctured sphere. 
T. Budd's suggestion that the total volumes were identical remained very puzzling to us (and somehow paradoxical), in view of the local form \ref{mupsiN} obtained in \cite{DavidEynard2014} for the angle measure in terms of first Chern classes $\psi_v$ only, until we became aware of the discontinuity phenomenon discussed in sec.~\ref{ssDiscont}. 
Indeed, according to \cite{Wolpert:1990} the Weil-Petersson local form for the volume form equals the Mumford $\kappa_1$-class, and cannot be expressed only in term of the Chern classes $\psi_v$ given by \ref{psivdef}. We shall come back to the relation and differences between our model and the usual Witten-Konsevitch topological intersection theory based on general $\psi_k$ Chern classes in the subsection \ref{sssIntersection}.

We have shown here that there is in fact no contradiction. Locally the K\"ahler 2-form for our model, and the Weil-Petersson 2 form are proportional. This is a much stronger result.
It implies of course that the total volume of $\mathfrak{D}_{N}$ of our model (volume of Delaunay triangulations with $Ns$ points) and the Weil-Petersson volume of $\overline{\mathcal{M}}_{0,N}$ indeed coincide. 

\paragraph{Higher genus :}
Although this has not been presented in details, let us 
note that it is possible to generalize the random Delaunay model from the planar case (genus $g=0$) to the higher genus $g>0$ case. One simply has to consider triangulations $T$ of genus $g$, with angles variables $\theta_e$ associated to the edges $e$, to use the flat measure on the angles, and to take into account the local and global constraints over the angles (this is the non-trivial part).

Since the identification \ref{Del=WP} between the Delaunay K\"ahler form and the Weil-Petersson form is local, it should also be valid for the $g>0$ case.

\paragraph{Delaunay mapping versus uniformization mapping of $\overline{\mathcal{M}}_{0,n}$ :}
Let us stress that our parametrization of the the moduli space $\overline{\mathcal{M}}_{0,n}$ of the punctured sphere by Delaunay triangulations is different from the parametrization for $\overline{\mathcal{M}}_{0,n}$ used in \cite{TakhZograf:1987}-\cite{Zograf:1993}, although both involve $n$ punctures coordinates $z_i$, ${i=1,n}$ in the complex plane. 
Indeed in our model the resulting constant curvature metric in $\mathbb{C}$ (out of the punctures) for a representative of an element in $\overline{\mathcal{M}}_{0,n}$ is implemented by gluing triangles endoved with the non-conformal Beltrami-Cayley-Klein metric. In \cite{TakhZograf:1987} \cite{Zograf:1993} the constant curvature metric in $\mathbb{C}$ (out of the punctures) is conformal, and is given by a classical solution of Liouville equation. This explains why the K\"ahler prepotential for the Weil-Petersson metric takes a different form in our model (where it is a hyperbolic volume) and in their model (where it is given by a Liouville action evaluated at a classical solution which corresponds to the constant curvature metric.

Let us also recall that in our parametrization, the constant curvature metric is in fact smooth along the triangle edges, since one glues hyperbolic triangles along geodesic borders (see e.g. \cite{Penner2006}).
 
The random Delaunay model remains an interesting model of random two-dimensional geometry since it is an explicit 
model of a global conformal mapping of an abstract (or intrinsic) but continuous two-dimensional geometry model onto the complex plane. 
This mapping through Delaunay triangulations is different, and somehow simpler, than the general mapping provided by the Riemann uniformization theorem, which is usually considered.
Indeed a local modification of the position of one vertex of the triangulation translates in a local modification of the associated K\"ahler form, since the K\"ahler potential $\mathcal{A}_T$ given by \ref{AsumV} is a sum over local terms (the hyperbolic volumes $\mathbf{V}(f)$ of the triangles).
This is not the case for the uniformization mapping, which leads to a global K\"ahler potential (a classical Liouville action).

Therefore the model discussed here should allow to study the local properties of the conformal mapping of a random metric onto the plane. We present new, although preliminary,  local results in the next section \ref{sLocIneq}.

\subsubsection{Relation between the model and $c=0$ pure gravity ($(3,2)$ Liouville)}

This identity between the Weil-Petersson 2-form and our K\"ahler form on the space of random triangulations shows that the random Delaunay triangulation model is equivalent to the abstract topological model based on the Weil-Petersson measure on moduli spaces. This implies that the Random Delaunay Triangulation model is in the universality class of pure two-dimensional gravity (Liouville theory with $\gamma=\sqrt{8/3}$ and $c_\mathrm{matter}=0$). This was conjectured in \cite{DavidEynard2014}, but on heuristic arguments. 
We now discuss precisely the relationship between the random Delaunay triangulations model and $c=0$ pure 2d gravity.

The equality \ref{Del=WP} shows that the measure $d\mu (\mathbf{z})$ on $\mathbb{C}^D$ is actually the Weil-Petersson volume form on $\mathcal{M}_{0,N+3}$. Indeed, we have shown that the measure on Delaunay triangulations having $N+3$ vertices is
\beq
d\mu(\mathbf{z})=2^N \det\left[D_{u\bar{v}}\right]\prod\limits_{v=4}^{N+3}dz_v^2
=\frac{2^N}{N!}\Omega_{\mathscr{D}}^N
= \frac{1}{N!}\Omega_{\mathscr{W\!\!P}}^N.
\eeq
The r.h.s. is precisely the Weil-Petersson volume form on the moduli space ${\mathcal{M}}_{0,N+3}$.
Hence we have for the total volumes
\begin{equation}
\label{ }
\Vol_{\mu}(\widetilde{\mathcal{T}}_{N+3}^f) = \mathrm{Vol}\left(\overline{\mathcal{M}}_{0,N+3}\right)
\end{equation}

The total volumes of the compactified moduli space $\overline{\mathcal{M}}_{g,n}$ of Riemann surfaces of genus $g$ with $n$ punctures are known explicitely. They satisfy recursion relations which have been studied by \cite{Zograf:1993} \cite{KMZ:1996}, \cite{MZ:2000}.
When discussing the relation with $c=0$ pure 2d gravity, it is sufficient to consider the large $n$ behaviour of $\mathrm{Vol}\left(\overline{\mathcal{M}}_{g,n}\right)$. 
It behaves for large $n$ as (see Eq.~B.1 in \cite{MZ:2000})
\begin{equation}
\label{VolMgnExpl}
\mathrm{Vol}\left(\overline{\mathcal{M}}_{g,n}\right)\  =\  {\mathbf{C}}^{^{\scriptstyle{n}} }\, n^{(5g-7)/2} \left(\mathbf{a}_g+ \mathcal{O}(1/n)\right)
\end{equation}
$\mathbf{C}$ is a positive constant independent of $g$, and the $\mathbf{a}_g$ are known positive constants.
Note that in \ref{VolMgnExpl} we omit the $n!$ factor which is present in the explicit result of \cite{MZ:2000}. This $n!$ factor comes from the fact that in \cite{MZ:2000} the $n$ punctures are \emph{labelled}, while in our model the vertices, hence the punctures, are \emph{unlabelled}.
For genus $g=0$ this gives $\mathrm{Vol}\left(\overline{\mathcal{M}}_{0,n}\right) \sim {\mathbf{C}}^{^{\scriptstyle{n}} } n^{-7/2} {\mathbf{a}}_{_{0}}$

In order to characterize the universality class of the random Delaunay triangulation model, we must study the large $N$ limit of the volume of the space of Delaunay triangulations with $N+3$ points, modulo $\mathrm{SL}(2,\mathbb{C})$ transformations, i.e the triangulation spaces $\mathcal{\tilde T}^f_{N+3}$ defined by \ref{defTtilde}. 
The number of vertices $N+3$ is considered as the ``volume'' of the triangulations; by analogy with what is done in 2D gravity and in its random matrix model formulation.

Indeed in continuum 2d quantum gravity, a two dimensional manifold $M$ is embodied with a ``quantum'' fluctuating metric $h$.
Its volume (area) $A=\int d^2z\,\sqrt{h}$ is therefore not fixed.  A ``string tension'' or ``cosmological constant'' $\Lambda$ parameter is associated to the volume $A$, so that the ``partition function'' is defined by the functional integral over metrics
\begin{equation}
\label{ }
\mathcal{Z}(\Lambda)=\int \mathcal{D}[h]\ \exp\left(- \,\Lambda \int d^2z \,\sqrt{h}\right)
\end{equation} 
In the continuum formulation, KPZ scaling is known to imply that for the genus zero case (planar surfaces), $\mathcal{Z}(\Lambda)$ scales with $\Lambda$ as
\begin{equation}
\label{ZofLambda}
\mathcal{Z}(\Lambda) \propto\ {\Lambda}^{2-\gamma_s}
\quad\text{with}\quad \gamma_s=-1/2  \quad\text{the string exponent of pure gravity.}
\end{equation} 
$\gamma_s$ is the string exponent.
In its discretized version (random maps), the partition function is defined by a sum over (for instance) abstract planar triangulations $T$ (with equal length $\ell$ for all edges), and denoting by $|T|$ the number of triangles of $T$, by
\begin{equation}
\label{ZofZsing}
\mathcal{Z}(Z)=\sum_{\text{triangulations}\ T} Z^{|T|}
\end{equation} 
$Z$ is a fugacity associated to the triangles, and we omit the possible symmetry factors.
Combinatorics and random matrix methods show that this partition function is an analytic function of the fugacity $Z$ near the origin, with a first singularity at some critical $Z_c>0$, of the form (in the planar case)
\begin{equation}
\label{ }
\mathcal{Z}(Z)\ =\ \text{regular part}\ +\ \mathbf{A}\,(Z_c-Z)^{5/2}\ +\ \text{subdominant terms}
\end{equation}
This allows the identification between the ``renormalized chemical potential'' $z_0$ for the triangles and the continuum cosmological tension 
\begin{equation}
\label{ }
{(Z_c-Z)\over Z_c}\, =\ \ell^2 z_0 \,=\, \ell^2\Lambda
\end{equation}
The continuum limit where $\ell\to 0$ while $z_0=\Lambda$ is fixed, amounts to consider the limit when $Z\to Z_c$, and to consider only triangulations with large size or area $|T|$. The exponents $5/2$ in \ref{ZofLambda} and \ref{ZofZsing} coincide, hence the identification of continuum $c=0$ pure gravity with the continuum limit of the random maps model.

The construction of a continuum limit from the discrete Delaunay triangulation model is done in exactly the same way.
We consider the grand canonical ensemble of planar Delaunay triangulations $\widetilde{\mathcal{T}}^f=\cup_n\widetilde{\mathcal{T}}^f_{n}$,such that $n$ varies (from $n=3$ to $\infty$), and associate a ``fugacity''$Z$ to each vertex of the triangulations.
The partition function is 
\begin{equation}
\label{ }
\mathcal{Z}_{\mathscr{D}}(Z)=\sum_n Z^n\  \mathrm{Vol}(\widetilde{\mathcal{T}}^f_{n})
\end{equation}
From the explicit results of \cite{Zograf:1993} \cite{KMZ:1996}, \cite{MZ:2000}
this function is analytic in a neighborhood of zero, with its  closest singularity at $Z_c= \mathbf{C}^{-1}$.
From \ref{VolMgnExpl} the singularity is again of the form
\begin{equation}
\label{ }
\mathcal{Z}_{\mathscr{D}}(Z)\ =\ \text{regular part}\ +\ A\,(Z_c-Z)^{5/2}\ +\ \text{subdominant terms}
\end{equation}
with the same exponent $2-\gamma_S=5/2$.
It is this fact that leads to the statement that the continuous limit (large number of vertices) of Delaunay triangulations is a (3,2) minimal model, which has central charge $c=0$, dressed by gravity. 
Let us stress that the Delaunay triangulations are considered as models of random discrete geometry, not as models of continuous geometries anymore.
We shall come back to a more precise discussion of this difference in \ref{AbContLim}.

\subsubsection{Relationship and difference with $c=-2$ topological gravity}
\label{sssIntersection}

Our model can be related to the Witten-Kontsevich intersection theory as follows:
the measure $d\mu$ is expressed in terms of the Weil-Petersson two-form $\Omega_{\mathscr{WP}}$, and which is proportional to the Mumford $\kappa_1$-class:
\beq
\Omega_{\mathscr{WP}}=2\pi^2 \kappa_1.
\eeq
Therefore, the volume of $\widetilde{\mathcal{T}}_{N+3}^f$ measured with $d\mu$ is the Weil-Petersson volume of $\overline{\mathcal{M}}_{0,N+3}$:
\bea
\Vol_{\mu} (\widetilde{\mathcal{T}}_{N+3}^f)
&=&\Vol_{\mathscr{WP}} (\overline{\mathcal{M}}_{0,N+3})=\frac{1}{N!}\int_{\overline{\mathcal{M}}_{0,N+3}}\Omega_{\mathscr{WP}}\cr
&=&\int_{\overline{\mathcal{M}}_{0,N+3}}\frac{(2\pi^2\kappa_1)^N}{N!}\cr
&=&\int_{\overline{\mathcal{M}}_{0,N+3}} e^{2\pi^2 \kappa_1} \cr
&=&\left\langle e^{2\pi^2\kappa_1} \tau_0^{N+3}\right\rangle_{0},
\eea 
where the 3rd equality simply means that, when one performs the expansion of the exponential in series of $\kappa_1$, the only cohomological class which  has the right dimension when integrated over $\overline{\mathcal{M}}_{0,N+3}$ is $\frac{(2\pi^2 \kappa_1)^N}{N!}$, and the last equality is the standard Witten notation for the intersection number of $\kappa_1$ class.\\ 
The Weil-Petersson volume is a special case of the volume $\left\langle \tau_0^{N+3} \ e^{\sum\limits_{k\geq0}\hat{t}_k\kappa_k}\right\rangle_{0}$ with the times $\hat{t}_k=\delta_{1,k}2\pi^2$. The integral of $\kappa$-classes can be computed in terms of intersection numbers of Chern classes $\tau_d$. The formula relating those intersection numbers is the following:
\beq 
\label{kappa-chern}
\left\langle \prod_{i=1}^{N+3} \tau_{d_i} \ e^{\sum\limits_{k=0}^\infty \hat{t}_k \kappa_k} \right\rangle_{0}=\left\langle \prod_{i=1}^{N+3} \tau_{d_i} \ e^{\sum\limits_{k=0}^{\infty}(2k+1)!!t_{2k+3}\tau_{k+1}} \right\rangle_{0},
\eeq 
(in the right hand side, one has to write the Taylor expansion of the exponential, and the number of marked points is read off from the number of $\tau$-factors)
where the times $\hat{t}_k$ and $t_k$ are related as the coefficient of the series in powers of $u$:
\beq 
e^{-\sum\limits_{k\geq0}\hat{t}_k u^{-k}}=1-\sum\limits_{k=0}^{\infty}(2k+1)!!t_{2k+3}u^{-k}.
\eeq
The generating function of all intersection numbers (not only genus 0) is the  KdV Tau function
\beq
\left\langle e^{\sum\limits_{k=1}^{\infty} (2k-1)!! t_{2k+1}\tau_{k}} \right\rangle
= \log\left( \tau_{\text{KdV}}(0,t_3,3t_5,15t_7,\dots,(2k-1)!! t_{2k+1},\dots) \right),
\eeq
In \cite{Gaiotto-Rastelli}, this is interepreted as topological gravity, with the times $t_k$ associated to closed FZZT stable branes. The background $t_k=0$ is the $(2,1)$ theory, which has central charge $c=-2$, and \cite{Gaiotto-Rastelli} mention that other background times are associated to other $(p,q)$ minimal models with central charge $c=1-6(p-q)^2/pq$.
In our case, the background times are 
\beq
t_{2k+3} = \delta_{k,0}-\frac{(-1)^{k} \ (2\pi)^{2k}}{(2k+1)!},
\eeq
which are associated to the $(3,2)$ model with $c=0$, as shown by \cite{KMZ:1996}.

Therefore, for fixed $N$, the volume of the set of Delaunay triangulations of size $N+3$ is deduced from the KdV tau function specialized at the times $t_{2k+3}$, that are not null. This is the difference with $c=-2$ topological Witten-Kontsevich, where the times are closed to $0$, and where the correlation functions are logarithmic derivatives of the KdV tau function evaluated at null times.\\

\section{Local inequalities on the measure}
\label{sLocIneq}

We have shown that the measure $\mu=d\mu(\mathbf{z})$ on $\mathcal{T}_{N}$ can be expressed locally in terms of the Weil-Petersson measure on ${\overline{\mathcal{M}}}_{0,N}$, and that it is necessary to consider its large $N$ limit to make contact with 2d quantum gravity.
One would like to study the convergence of the sequence of random measure on points in the complex plane (a random point process in $\mathbb{C}$, and of its correlations functions (moments and cumulants) defined from the random Delaunay triangulations ensembles. 
It is thus necessary to study the \emph{local properties} in $\mathbb{C}$ of this random measure, not only the properties of the total volume of the set of random measure, as discussed before.
A first step towards convergence results is to obtain local bounds and inequalities between the different measures.
The ultimate goal is to study the continuous $N\to\infty$ limit. Since the concept of ``continuum limit'' may have several different mathematical meanings, let us briefly come back what we mean in this study in \ref{AbContLim}.
We then give two local results satisfied by the measure $\mu(T,d\boldsymbol{\theta})\ =\ d\mu(\mathbf{z})$ defined  in \cite{DavidEynard2014} and studied in this paper.

\subsection{About the continuum limits}
\label{AbContLim}

\paragraph{Random metrics:} One may first consider our model as a model of \emph{random Riemannian metrics} on the plane. Indeed, if one looks at a configuration of distinct points on the Riemann sphere and adds progressively points, we have seen in the previous part that at each step, one can embody each triangle of the Delaunay triangulation with its natural Beltrami hyperbolic metric. This was a useful tool in Section~\ref{sRelWeilPet}. This gives a global metric which is hyperbolic but has puncture singularities at the points. Since the punctures are at infinite distances, and the metric is not compact, there is clearly no hope, even with an ad hoc rescaling of the metric, that the space of triangulations equipped with the Beltrami metric has a limit in the Gromov-Hausdorff sense when $N\to\infty$. 

However, by analogy with the random planar map model, one should rather consider random \emph{discrete metrics spaces} constructed from the random triangulations. 
One may consider for instance the random Voronoï graph $T^\star$, dual to the random Delaunay triangulation $T$, whose vertices are the centers (with coordinates $w_f$) of the circumcircles to the faces (the triangles) $f$ of the triangulation $T$. An edge $e^\star=(w_f,w_{f'})$ of $T^\star$ is dual to an edge of $T$, since it is the straight segment between the centers of two faces $f$ and $f'$ adjacents to an edge $e=(z_{v_1},z_{v_2})$ of $T$ (with the notations of Fig.~\ref{definition_angles}). This dual edge $e^\star$ is in fact a geodesic in the global hyperbolic Beltrami metric, with length
\begin{equation}
\label{ }
l(e^\star)= {1\over 2} \log\left({(1+\sin(\theta_n))(1+\sin(\theta_s))\over (1-\sin(\theta_n))(1-\sin(\theta_s))}\right)
\end{equation}
where $\theta_n=\mathrm{Arg}((w_f-z_{v1})/(z_{v_2}-z_{v_1}))$ is the angle between the vectors $(v_1,v_2)$ and $(v_1,f)$
while $\theta_s=\mathrm{Arg}((z_{v_2}-z_{v_1})/ w_{f'}-z_{v_1}))$ is the angle between the vectors $(v_1,f')$ and $(v_1,v_2)$.
This defines a distance function $d_{T^\star}$ on the Voronoi graph $T^\star$, and this mapping $T^\star \to d_{T^\star}$ is continuous on the space of Delaunay triangulations with $N$ points, since it is continuous when one performs a flip. Indeed an edge $e$ is flipped if $\theta(e)=\theta_n+\theta_s=0$, hence when $l(e^\star)=0$.
A natural conjecture, is that the large $N$ limit makes sense and is in the same universality call of the random planar map model, namely that $(T^\star,N^{-1/4} d_{T^\star})$, considered as a random (discrete) metric space, converge in the Gromov-Hausdorff sense towards the Brownian map as in LeGall \cite{LeGall2013} and Miermont \cite{Miermont2013}.

Note that other choices are expected to have the same convergence properties.
For instance one can choose $l(e^\star)=2 \sin((\theta_n+\theta_s)/2)$. This corresponds to the locally flat ``intrinsic planar triangulations'' case considered in the section~2.7.1 of \cite{DavidEynard2014}, and which is an example of ``critical discrete conformal map'' in the sense of \cite{MercatHBTT}.

\paragraph{Random conformal measures:} Secondly, one may consider our model as a model of \emph{random measures} on the plane. This is our point of view in this work, and this is the one relevant when discussing the relation between the continuum limit of our model, or of random maps, with the Liouville QFT. Indeed in the Liouville theory the Liouville field $\phi_L$ defines by its exponential $\exp(\gamma \phi_L)$  random measures with fascinating multifractal properties linked to multiplicative chaos theory (see for instance \cite{DKRV2015} and references therein), and conformal invariance. It is for instance expected that the moments of the \emph{local density} of points $\rho(z)^\beta$ are related to the local vertex operators $\exp(\alpha \phi_L(z))$ in the Liouville theory.
This point of view is deeply connected to the field of random conformal geometry and conformal stochastic processes (such as SLE, CLE, QLE) and the study of their local properties.
In this direction, one must address already non-trivial questions, many still open, which are (as far as we know) out of the reach of topological theories and TQFT methods.

\subsection{Maximality property over the Delaunay triangulations}
\label{sMaxPropMes}

Looking at the measure $d\mu(\mathbf{z})$ on $\mathbb{C}^D$ (the space of distributions of $N+3$ points on the Riemann sphere), theorem \ref{themeasurez} gives $$d\mu(\mathbf{z})=\prod_{v=4}^{N+3} d^2z_v\ 2^N\ \det\left[{D_{u\bar v}}\right]_{u,v\notin\{ z_1,z_2,z_3\}}.$$
Let $\{z_v\}$ be a configuration of $N+3$ points on the Riemann sphere, and let $T$ be a planar triangulation associated to these points. Then the K\"ahler metric $D_{u\bar v}(T) $ on $\mathbb{C}^{N+3}$ is still well defined ($T$ is not necessarily a Delaunay triangulation), see equations \ref{Duvbar} and \ref{AsumV}. We use a short-hand notation:
\begin{equation}
d_{(ijk)}(T)=\det{\left[ (D_{u \bar v})_{u,v \neq \{ i,j,k\}} (T)\right]}.
\end{equation}
Then the following result holds:
\begin{theorem}
\label{propmaximal}
Given $N+3$ points $z_{1},\ldots,z_{N+3}$ in $\mathbb C$, their Delaunay triangulation $T^{D}(\{z_v\})$ is the one which maximizes $ d_{(ijk)}(T)$ among all possible triangulations $T$:
\begin{equation}
d_{(ijk)}(T^{D}(\{z_v\}))=\max_{T\, {\rm triangulation}\,{\rm of}\, \{ z_{v}\}} d_{(ijk)}(T).
\end{equation}
\end{theorem}
In order to prove this assertion, one may look at the transformation $d_{(124)}(T)\xrightarrow[(13)]{(24)}d_{(124)}(T') $ where the triangulation $T$ undergoes the edge flip $(24)\to (13)$ (see figure \ref{lawson}). It leads to the following lemma:
\begin{lemma} 
\label{areadet}
Denote $f$ the triangle $(124)$, and $\omega_f$, $R_f$ respectively the center and the radius of its circumcircle. Then 
\begin{equation}
d_{(124)}(T)-d_{(124)}(T')=\det{\left[ (D_{u \bar v})_{u,v \neq \{ 1,2,3,4\} } (T)\right]}\times{\rm Area}(f) \frac{\left| z_{3}-\omega_{f}\right|^{2}-{R_{f}}^{2}}{\left|z_3-z_1\right|^{2}\left|z_3 - z_2\right|^{2}\left|z_3 - z_4\right|^{2}}.
\end{equation}
\end{lemma}
\begin{proof}
The proof of this lemma is given in appendix \ref{appA}. 
\end{proof}
\begin{remark}
Let us recall that in \cite{DavidEynard2014}, for a triangulation where all the faces are positively oriented, it has been proved that the Hermitian form $D_{u\bar v}$ is positive. The result is true for general planar triangulations, if we impose a positive orientation for the faces (orientation that we enforce here). Hence, the principal minors of $D_{u\bar v}$ are positive, so $ \det{\left[ (D_{u \bar v})_{u,v \neq \{ 1,2,3,4\} } (T)\right]}\geq 0$.
\end{remark}
\begin{remark}
From lemma \ref{areadet}, one deduces that $d_{(124)}(T)-d_{(124)}(T')\geq0$ only if $z_3$ is out of the circumcircle of $f$.
\end{remark}

In \cite{DavidEynard2014}, the authors proved that the quantity $d_{ijk}(T)$ changes covariantly when changing the points $(ijk)$, a useful property for demonstrating theorem \ref{propmaximal}:
\begin{lemma}
\label{covariance}
The quantity 
$$
\frac{
{d_{(ijk)}(T)}
}{ {\left| \Delta_{3} (i,j,k) \right| }^{2}} \ ,
$$ with
${\Delta_{3}} (i,j,k)=( z_{i}-z_{j})( z_{i}-z_{k})( z_{j}-z_{k})$, is independent of the choice of the three fixed points $\{ z_{i}, z_{j}, z_{k} \}$.
\end{lemma}

\begin{proof}[Proof of the theorem \ref{propmaximal}]
Take a triangulation $T$ of the configuration $\{z_v\}$, then apply the recursive Lawson flip algorithm (see \cite{Brevilliers2008} or \cite{Lawson1972} for details on this algorithm, noted here LFA) to $T$, one obtains the Delaunay triangulation $T^{D}(\{z_{v}\})$ of the $\{z_{v}\}$. At each step, the LFA applies a single edge flip. Note $(T_{i})_{0\leq i \leq n}$ the sequence of successive triangulations obtained by the LFA, with $T_{0}=T$ and $T_{n}=T^{D}(\{z_{v}\})$:
\begin{equation}
T_{0}=T\xrightarrow[(a^1 d^1)]{(b^0 c^0)}T_{1}\xrightarrow[(a^2 d^2)]{(b^1 c^1)}\dots\xrightarrow[(a^n d^n)]{(b^{n-1} c^{n-1})}T_{n}=T^{D}(\{z_{v}\})
\end{equation}
The LFA works in the following way: for $T_{i}$, if it is not Delaunay, at least one point, say $a^{i}$, is contained in the circumcircle of a neighboring face $(b^{i},c^{i},d^{i})$. The situation is depicted in figure \ref{lawson}. \\
\begin{figure}[!h]
\centering
\includegraphics[width=14cm]{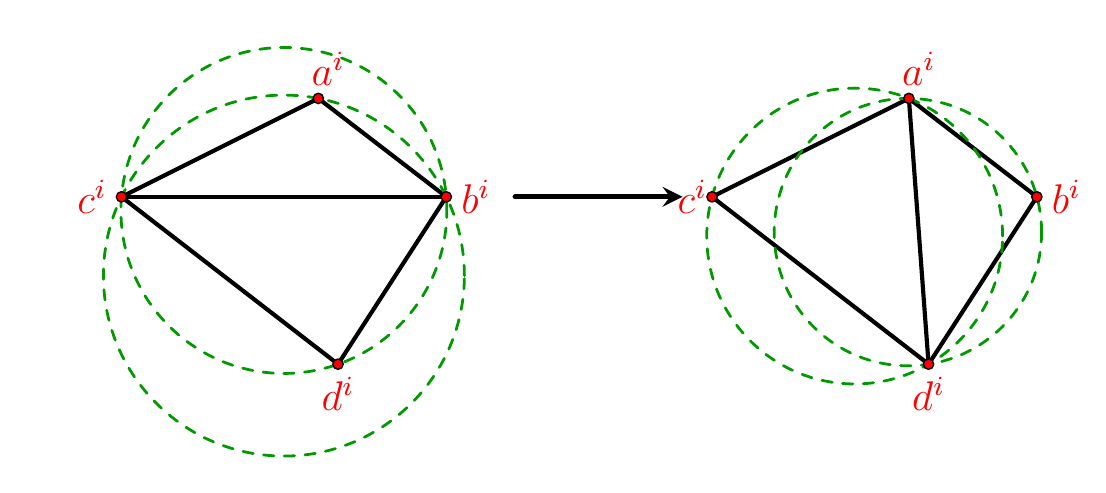}
\caption{Effect of a flip at one step of the LFA.}
\label{lawson}
\end{figure}

Then the edge $(b^{i},c^{i})$ is flipped to give $(a^{i},d^{i})$. It follows that for the two new faces $(a^{i},c^{i},d^{i})$ and $(a^{i},d^{i},b^{i})$, their circumcircles enclose respectively  neither $b^{i}$ nor $c^{i}$. Using lemma \ref{areadet} and \ref{covariance}:
\begin{eqnarray*}
d_{(123)}(T^{D})-d_{(123)}(T) & = & \sum_{i=0}^{n-1}\left[  d_{(123)}(T_{i+1})-d_{(123)}(T_{i}) \right]\\
& = & \left| {\Delta_{3}} (1,2,3) \right|^{2}\sum_{i=0}^{n-1}\left[ \frac{d_{( a^{i},c^{i},d^{i})}(T_{i+1})-d_{( a^{i},c^{i},d^{i})}(T_{i}) }{\left| {\Delta_{3}} (a^{i},c^{i},d^{i}) \right|^{2}}\right]\\
& \geq & 0
\end{eqnarray*}
which ends the proof.
\end{proof}

The measure $d\mu(\mathbf{z})$ used is then maximal over the Delaunay triangulations.  

\subsection{Growth of the volume}
\label{ssGrowVol}
The second result relates to the N dependence of the total volume 
\begin{equation}
V_N = \int_{{\mathbb C}^{N}}\prod_{v=4}^{N+3}d^{2}z_{v} 2^N \det{\left[D_{u\bar v}(T^{D}(\{z_v\}))\right]}_{u,v\notin\{ z_1,z_2,z_3\} }
\end{equation}
It is the volume of the space of Delaunay triangulations with $N+3$ vertices with the measure $d\mu(\mathbf{z})$. A  lower bound of the growth of the volume when the number of vertices increases is given by the following inequalities:
\begin{theorem}
If we add a $N+4$'th point to a given triangulation and integrate over its position, the following inequality holds:
\begin{eqnarray}
\label{ineqvolume}
    \int_{\mathbb{C}}d^2 z_{N+4} \det{\left[D_{u\bar v}(T^{D}(\{z_1,\dots,z_{N+4}\}))\right]}_{u,v\notin\{ z_1,z_2,z_3\} }\cr
    \geq (N+1)\frac{\pi^2}{8} \det{\left[D_{u\bar v}(T^{D}(\{z_1,\dots,z_{N+3}\}))\right]}_{u,v\notin\{ z_1,z_2,z_3\} }
    \end{eqnarray}
It implies the inequality for the total volumes
\begin{equation}
V_{N+1}\geq (N+1)\frac{\pi^2}{8}V_{N}.
\end{equation}
\end{theorem}

Before proving the theorem, let us stress that this growth property is global w.r.t. the last point. 
A similar inequality does not stand locally for the measure $\det{\left[D_{u\bar v}(T^{D}(\{z_v\}))\right]}_{u,v\notin\{ z_1,z_2,z_3\} }$ when one adds a vertex at a fixed position to an existing Delaunay triangulation.
This has been checked numerically.\\

\begin{proof}
We first focus on the inequality \ref{ineqvolume}. The proof follows the following procedure:
    \begin{itemize}
        \item Fix $N+3$ points $\{z_1 , \dots,z_{N+3}\}$ in $\mathbb{C}$, and note $T^D (\{z_v\})$ the Delaunay triangulation constructed on this configuration.
        \item Pave the Riemann sphere with regions $\mathcal{R}(f)$ (defined bellow) associated with the faces $f$ of the triangulation.
        \item Then add a point $z_{N+4}$ in $\mathbb{C}$ to this triangulation. Depending on the region $\mathcal{R}(f)$ where it stands, transform the triangulation to include the new point and compute the measure associated with this triangulation.
        \item Integrate over $z_{N+4}$, find a  lower bound on of the integral, and compare the result with the measure associated with $T^D (\{z_v\})$.
    \end{itemize}
    For the Delaunay triangulation constructed over $\{z_1,\dots z_{N+3}\}$, the Riemann sphere can be conformally paved with regions $\mathcal{R}(f)$ associated to each face in the following way. Let us look at the edge $e$ whose neighboring faces are $f$ and $f'$. The circumcircles of $f$ and $f'$ meet at the vertices located at the ends of $e$ with an angle $\theta'(e)=(\pi-\theta(e))$. Define $\mathcal{C}_e$ the arc of a circle joining the ends of $e$, and making an angle $\theta'(e)/2={(\pi -\theta(e))}/{2}$ with each of the circumcircles of $f$ and $f'$ at the vertices of $e$. See figure \ref{arc_bissecteur}. The region $\mathcal{R}(f)$ is now defined as the domain enclosed in the three arcs of a circle $\mathcal{C}_{e_1}$, $\mathcal{C}_{e_2}$, $\mathcal{C}_{e_3}$ corresponding to the three edges $e_1$, $e_2$, $e_3$ surrounding $f$ (see figure \ref{def_regR}). This domain is now transformed covariantly under a M\"obius transformation.\\
    \\
    \begin{figure}[!h]
    \centering
    \includegraphics[width=8cm]{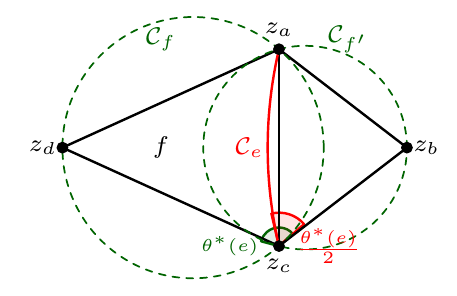}
    \caption{Definition of the arc $\mathcal{C}_e$.}
    \label{arc_bissecteur}
    \end{figure}
    We add the point $z_{N+4}$ in the Riemann sphere. If $z_{N+4}\in\mathcal{R}(f)$, we construct the triangulation $T_f^D(\{z_1,\dots,z_{N+3}\},z_{N+4})$ by joining the vertex $z_{N+4}$ to the vertices $a$, $b$ and $c$ of the face $f$. The triangulation  $T_f^D(\{z_1,\dots,z_{N+3}\},z_{N+4})$ is in general different from the Delaunay triangulation $T^D(\{z_1,\dots,z_{N+4}\})$. Yet it is still possible to define the measure $\det{D_{u\bar v}(T_f^D (\{z_v\},z_{N+4}))}$, which is still a positive quantity, and which, from theorem \ref{propmaximal}, satisfies:
    \begin{equation}
    \det{\left[D_{u\bar v}(T_f^D (\{z_v\},z_{N+4}))\right]}_{u,v\notin\{ z_1,z_2,z_3\}}\leq\det{\left[D_{u\bar v}(T^D (\{z_1,\dots,z_{N+4}\}))\right]}_{u,v\notin\{ z_1,z_2,z_3\}}
    \end{equation}

    The aim is to find a  lower bound to the integral over each region $\mathcal{R}(f)$. The interesting result is that we found a  lower bound that does not depend on the region, although the shapes of the regions depend on the angle $\theta(e)$ between two neighboring circumcircles. We take this dependence out by integrating over smaller regions $\mathcal{B}(f)\subseteq \mathcal{R}(f)$. for the face $f$, $\mathcal{B}(f)$ is the region enclosed by the three arcs of circle that pass through two of the vertices of $f$ and that are orthogonal to the circumcircle of $f$ (see figure \ref{def_regB}).\\
    
    \begin{figure}
    \minipage{0.4\textwidth}
        \includegraphics[width=\linewidth]{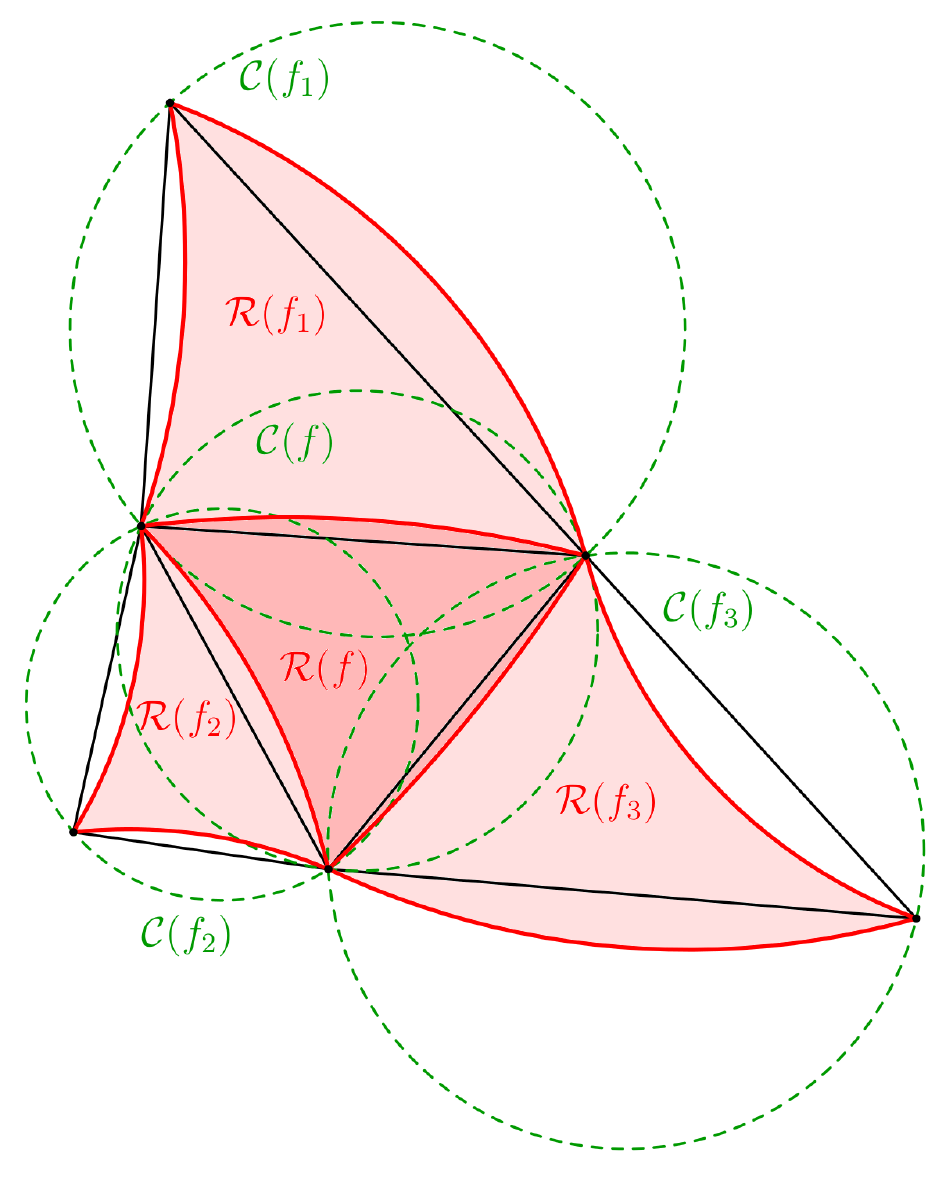}
        \caption{The region $\mathcal{R}(f)$ is enclosed in the bissector arcs.}
        \label{def_regR}
    \endminipage\hfill
    \minipage{0.4\textwidth}
    \includegraphics[width=\linewidth]{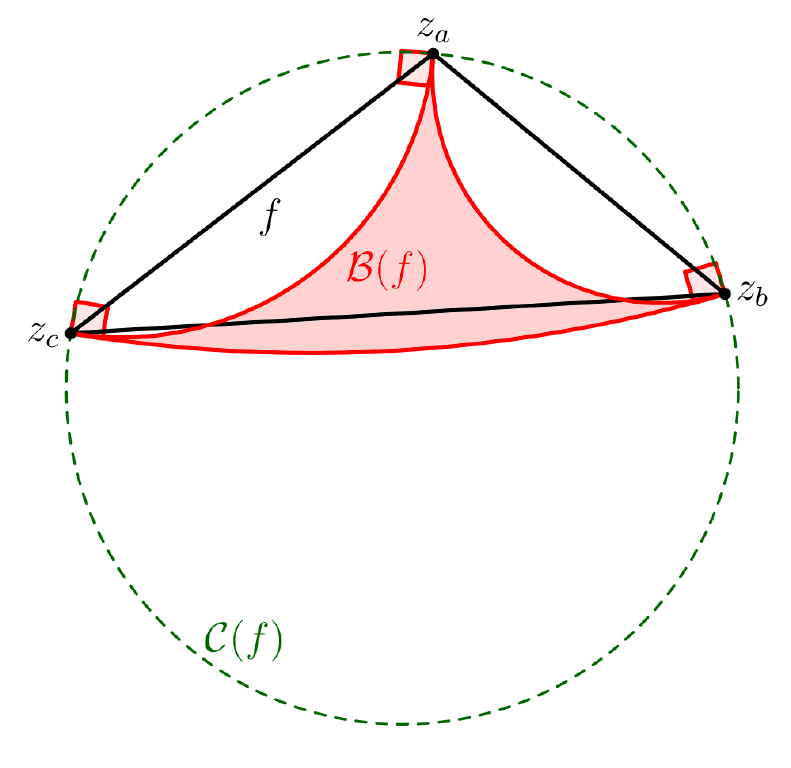}
    \caption{The region $\mathcal{B}(f)$ associated with a face $f$.}
    \label{def_regB}
    \endminipage
    \end{figure}

    The integration over $z_{N+4}$ thus decomposes in the following way:
    \begin{align}
    \int_{\mathbb{C}}d^2 z_{N+4} \, & \det{\left[D_{u\bar v}(T^{D}(\{z_1,\dots,z_{N+4}\}))\right]}_{u,v\notin\{ z_1,z_2,z_3\} }
    \nonumber\\
    &=\sum_{f}\int_{\mathcal{R}(f)}d^2 z_{N+4} \det{\left[D_{u\bar v}(T^{D}(\{z_1,\dots,z_{N+4}\}))\right]}_{u,v\notin\{ z_1,z_2,z_3\}}  \nonumber\\
    &\geq \sum_{f}\int_{\mathcal{R}(f)}d^2 z_{N+4} \det{\left[D_{u\bar v}(T_f^{D}(\{z_v\},z_{N+4}))\right]}_{u,v\notin\{ z_1,z_2,z_3\}}  \nonumber\\
    &\geq \sum_{f}\int_{\mathcal{B}(f)}d^2 z_{N+4} \det{\left[D_{u\bar v}(T_f^{D}(\{z_v\},z_{N+4}))\right]}_{u,v\notin\{ z_1,z_2,z_3\}}
    \end{align}
    In the last line, the integral can be computed explicitly. If $z_{N+4}\in\mathcal{B}(f)$ with  $f=(abc)$, one can compute the integration on $\mathcal{B}(f)$ using lemma \ref{covariance}: 
    \begin{align}
    \int_{\mathcal{B}(f)}d^2 z_{N+4}\ & \det{\left[D_{u\bar v}(T_f^{D}(\{z_v\},z_{N+4}))\right]}_{u,v\notin\{ z_1,z_2,z_3\}}
    \nonumber\\
    &=\frac{\Delta_3 (z_1,z_2,z_3)}{\Delta_3 (a,b,c)}\int_{\mathcal{B}(f)}d^2 z_{N+4} \det{\left[D_{u\bar v}(T_f^{D}(\{z_v\},z_{N+4}))\right]}_{u,v\notin\{ a,b,c\}}.
    \end{align}
    Then the right term factorizes nicely thanks to the shape of the triangulation around $z_{N+4}$:
    
    \begin{align}
    \label{encours}
     & \int_{\mathcal{B}(f)} d^2 z_{N+4} \ \det{\left[D_{u\bar v}(T_f^{D}(\{z_1,\dots,z_{N+4}\}))\right]}_{u,v\notin\{ a,b,c\}}\nonumber\\
    &=\int_{\mathcal{B}(f)}\hskip -1.em d^2 z_{N+4} \det{\left[D_{u\bar v}(T^{D}(\{z_1,\dots,z_{N+3}\}))\right]}_{u,v\notin\{ a,b,c\}}\times\det{\left[D_{u\bar v}(T^{D}(\{a,b,c,z_{N+4}\}))\right]}_{u,v\notin\{ a,b,c\}}.
    \end{align}
    
In the integrand, the term depending on $z_{N+4}$ is the second determinant, so we need to estimate:
\begin{equation}
\label{first_min}
I=\int_{\mathcal{B}(f)}d^2 z_{N+4} \det{\left[D_{u\bar{v}}(T^{D}(\{a,b,c,z_{N+4}\}))\right]_{u,v\notin\{ a,b,c\}}}
\end{equation}
It is the integral of the measure on the Delaunay triangulation made of the 4 points $a$, $b$, $c$ and $z_{N+4}$, where $z_{N+4}$ crosses the region $\mathcal{B}(f)$ (see figure \ref{triang_base}). The integral is computable if one considers the measure in terms of the angles (see equation \ref{muT}). With the notations of figure \ref{triang_base}:
\begin{equation}
I=\frac{1}{2}\int_{z_{N+4}\in\mathcal{B}(f)}d\theta_1 d\theta_2
\end{equation}
(We used here a result of the article \cite{DavidEynard2014}, expressing the measure in term of a basis of angles. Here, the angles $\theta_1$ and $\theta_2$ are a basis of this triangulation). 
\begin{figure}
\centering
\includegraphics[width=13.5cm]{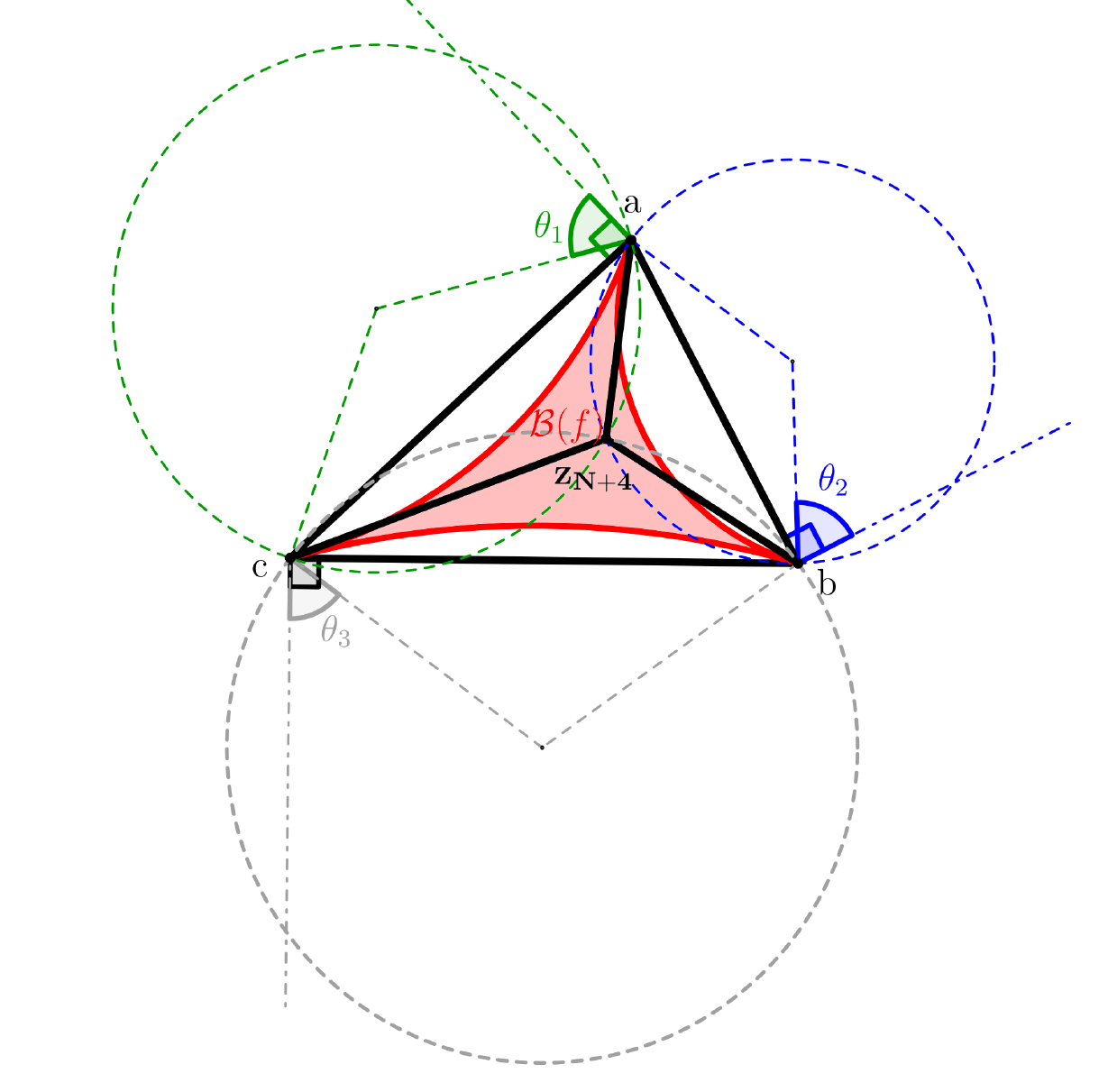}
\caption{The Delaunay triangulation (in black) with the associated circumcircles. The center of the external face is at $\infty$.}
\label{triang_base}
\end{figure}
The point $z_{N+4}$ belongs to the region $\mathcal{B}(f)$ if $ \theta_i^{\mathrm{min}}\leq \theta_i\leq \theta_i^{\mathrm{min}}+\frac{\pi}{2}$ for $i=1,2,3$. $\theta_i^{\mathrm{min}}$ corresponds to the angle $\theta_i$ for which the point $z_{N+4}$ is on the boundary arc of $\mathcal{B}(f)$ associated with the edge $i$.\\
We also have $\theta_1+\theta_2+\theta_3=\pi$ and $\theta_1^{\mathrm{min}}+\theta_2^{\mathrm{min}}+\theta_3^{\mathrm{min}}=\frac{\pi}{2}$, so eventually, $z_{N+4}\in\mathcal{B}(f)$ if:
\begin{align}
\label{ineq}
&\theta_1^{\mathrm{min}}\leq  \theta_1  \leq \theta_1^{\mathrm{min}}+\frac{\pi}{2}&\\
&\theta_2^{\mathrm{min}}\leq  \theta_2  \leq \theta_2^{\mathrm{min}}+\frac{\pi}{2}&\\
&\theta_1^{\mathrm{min}}+\theta_2^{\mathrm{min}}\leq  \theta_1 + \theta_2 \leq \theta_1^{\mathrm{min}}+\theta_2^{\mathrm{min}} +\frac{\pi}{2}&
\end{align}
From these conditions we immediately obtain that $I=\frac{1}{2}\left[\frac{1}{2}\left(\frac{\pi}{2}\right)^2\right] =\frac{\pi^2}{16}$. Then, one gets in equation \ref{encours}:
\begin{align}
\int_{\mathcal{B}(f)} d^2 z_{N+4}  \, & \det{\left[D_{u\bar v}(T_f^{D}(\{z_v\},z_{N+4}))\right]}_{u,v\notin\{ a,b,c\}}\nonumber\\
&=\frac{\pi^2}{16}\det{\left[D_{u\bar v}(T^{D}(\{z_v\}))\right]}_{u,v\notin\{ a,b,c\}}\nonumber\\
&=\frac{\pi^2}{16}\frac{\Delta_3 (a,b,c)}{\Delta_3 (z_1,z_2,z_3)}\det{\left[D_{u\bar v}(T^{D}(\{z_v\}))\right]}_{u,v\notin\{  z_1,z_2,z_3\}}
\end{align}\
So in the end:
\begin{align}
    \int_{\mathbb{C}}d^2 z_{N+4} \, & \det{\left[D_{u\bar v}(T^{D}(\{z_1,\dots,z_{N+4}\}))\right]}_{u,v\notin\{ z_1,z_2,z_3\} }
    \nonumber\\
&\geq \sum_{f}\frac{\pi^2}{16}\det{\left[D_{u\bar v}(T^{D}(\{z_v\}))\right]}_{u,v\notin\{  z_1,z_2,z_3\}}\nonumber\\
& \geq (N+1)\frac{\pi^2}{8}\det{\left[D_{u\bar v}(T^{D}(\{z_v\}))\right]}_{u,v\notin\{  z_1,z_2,z_3\}}
\end{align}
which gives the result:
\begin{equation}
V_{N+1}\geq (N+1)\frac{\pi^2}{8}V_{N}
\end{equation}
    
\end{proof}

The previous result gives a  lower bound which does not depend on the shape of the triangle by integrating over a restrained region $\mathcal{B}(f)$. If we do the same calculation and keep the region $\mathcal{R}(f)$, then the  lower bound is more accurate, but not universal any more. In this case, we then get a refined result:
\begin{theorem}
\label{refined}
\begin{align}
    \int_{\mathbb{C}}d^2 z_{N+4} \, & \det{\left[D_{u\bar v}(T^{D}(\{z_1,\dots,z_{N+4}\}))\right]}_{u,v\notin\{ z_1,z_2,z_3\} }
    \nonumber\\
& \geq \left[(N+1)\frac{\pi^2}{8}+\frac{1}{8}\sum_{e\in\mathcal{E}}\theta_e (2\pi-\theta_e)\right]\det{\left[D_{u\bar v}(T^{D}(\{z_v\}))\right]}_{u,v\notin\{  z_1,z_2,z_3\}}
\end{align}
\end{theorem}
(See appendix \ref{appB} for a proof).
We see that the angles associated to the triangulation appear. This angle-dependent term should be related to the kinetic term of the Liouville action in the continuum limit.

\section{Conclusion}
 It is believed that the quantum Liouville theory is the continuum limit of models of random maps equipped with a measure, as long as the models stand in the universality class of pure gravity. Although for a given finite $N$ (the number of vertices) each Delaunay triangulation allows to define a continuous metric over the  sphere (by gluing the Beltrami metrics on the triangles along their edges), therefore giving a global metric for each triangulation, the set of random metrics one looks at in the continuum limit is not this set of metrics given by triangulations of finite size $N$. The continuum limit means that one takes the number $N$ of vertices of the maps going to infinity $N\to\infty$, while keeping the maps random.
The goal is thus to understand how the probability space of Delaunay triangulations equipped with the measure given by the Lebesgue measure of circumcircle crossing angles, introduced in \cite{DavidEynard2014}, can converge (in the continuum limit) towards the quantum Liouville theory.
\medskip

In this article we continued to study the properties of this measure.

\begin{itemize}

\item In particular, we could compare the measure to the Weil-Petersson volume form. This allowed us to give an hyperbolic  representation for the measure.

\item We found an interesting property of maximality: our measure can be analytically continued to non-Delaunay triangulations, but is maximal exactly for Delaunay triangulations. This could open the possibility of some convexity properties, that need to be further explored, and that could be helpful in studying the continuum limit.

\item We found a lower bound on  the volume when one adds a new vertex $N\to N+1$. We have both local and global bounds. We could show that the partition function $V_N/N!$ grows at least like $(\pi^2/8)^N$. In other words, the log of the volume contains at least a term proportional to $N$, which can be interpreted as the "quantum area". 
If, as we expect, the continuum limit exists and is the Liouville theory at $c=0$ ($\gamma=\sqrt{8/3}$), then the log of the volume should converge towards the Liouville action.
The Liouville action is made of 2 terms ; one is the quantum area, the other is the kinetic energy.
The term we have found is compatible with a continuum limit of the quantum area.

It would be interesting to improve our bound, by taking into account the edges (integrals over $\mathcal R(f) -\mathcal B(f)$) contributions to see if they can account for the kinetic-term in the Liouville action.

\end{itemize}

All these results are encouraging steps towards a continuum limit that would be the Liouville pure gravity theory. 
Consequences need to be further explored.

\section*{Acknowledgements}

We are very grateful to Timothy Budd for his interest and for sharing his insights and the results of his unpublished calculations. F.D. is very much indebted to Jeanne Scott for her interest, her dedicated guidance in the mathematical literature and many lengthy and inspiring discussions. The results of section 3 owe much to her. We also thank G. Borot, L. Chekov, P. di Francesco and R. Kedem for fruitful discussions on the subject.

F.D. thanks the Perimeter Institute, U. of Iceland and Universidad de los Andes, Bogot\'a , where part of this work was done, for their hospitality.

B.E.'s work was supported by the ERC Starting Grant no. 335739 ``Quantum fields and knot homologies'' funded by the European Research Council under the European Union's Seventh Framework Programme, as well as the Centre de Recherches Math\'ematiques de Montr\'eal, the FQRNT grant from the Qu\'ebec government,
and the ANR grant "quantact".

\appendix
\section{Change of the measure with a flip: proof of lemma \ref{areadet}}\label{appA}
When the triangulation $T$ undergoes a flip to give the triangulation $T'$, only the two faces surrounding the edge change. So in the prepotentials $\mathcal{A}(T)$ and $\mathcal{A}(T')$, the only terms that differ are those implying the changed faces:
\begin{equation}
\mathcal{A}(T)-\mathcal{A}(T')=\mathbf{V}(124)+\mathbf{V}(234)-\mathbf{V}(123)-\mathbf{V}(134)
\end{equation}
Therefore, the differences between $D(T)$ and $D(T')$ are located in the $D_{i,j}$ with $i,j\in\{1,2,3,4\}$. As we are looking at the quantities $d_{(124)}$, the indices 1, 2 and 4 are not taken into account in the determinant. So the differences between $D(T)$ and $D(T')$ lay in $D_{3,3}$.
By expanding the determinant with respect to the third line, we get:
\begin{equation}
d_{(124)}(T)-d_{(124)}(T')=\left[D_{3,3}(T)-D_{3,3}(T')\right]\mathrm{det}\left[(D_{u\bar{v}})_{u,v\neq\{1,2,3,4\}}(T)\right]
\end{equation}
Let us focus on the term $D_{3,3}(T)-D_{3,3}(T')$. Using the form $D=\frac{1}{4i}AEA^\dagger$, and noting $z_{ij}=z_i-z_j$ one gets:
\begin{align}
D_{3,3}(T)-D_{3,3}(T') &= \frac{1}{4i}\left[\sum_{e\to3}\sum_{e'\mathrm{\, neighbour\, of \, }e} A_{3,e}E_{e,e'}\bar{A}_{3,e'}\right]&\\
&= \frac{1}{4i}\left[\frac{1}{z_{31}}\frac{-1}{\bar{z}_{32}}+\frac{1}{z_{32}}\frac{1}{\bar{z}_{31}}+\frac{1}{z_{31}}\frac{1}{\bar{z}_{34}}-\frac{1}{z_{34}}\frac{1}{\bar{z}_{31}}-\frac{1}{z_{32}}\frac{1}{\bar{z}_{34}}+\frac{1}{z_{34}}\frac{1}{\bar{z}_{32}}\right]&\\
&=\frac{1}{4i}\frac{z_{32}z_{34}\bar{z}_{31}\bar{z}_{42}+z_{31}z_{32}\bar{z}_{34}\bar{z}_{21}+z_{31}z_{34}\bar{z}_{32}\bar{z}_{14}}{|z_{31}|^2|z_{32}|^2|z_{34}|^2} &\\
&=\frac{1}{4i}\frac{N(z_3,\bar{z}_3)}{|z_{31}|^2|z_{32}|^2|z_{34}|^2}&
\end{align}
The coefficient of the term $z_3^2\bar{z}_3$ in $N(z_3,\bar{z}_3)$ gives $\bar{z}_{42}+\bar{z}_{21}+\bar{z}_{14}=0$. What is more, $\bar{N}(z)=-N(z)$, so $N$ can be written as  $N(z_3,\bar{z}_3)=a z_3\bar{z}_3+b z_3-\bar{b}\bar{z}_3+c$, with $a\in \mathrm{i}\mathbb{R}$, $b$ and $c\in \mathrm{i}\mathbb{R}$ functions of $z_i$, $\bar{z}_i$, $i=1,2,4$. Setting $\omega=-\frac{\bar{b}}{a}$ and $R^2=-\frac{c}{a}+|a|^2$,
\begin{equation}
N(z_3,\bar{z}_3)=a[(z_3-\omega)(\bar{z}_3-\bar{\omega})-R^2]
\end{equation}
$N(z_3,\bar{z}_3)=0$ is thus the equation of a circle for the point 3. As we have $N(z_i,\bar{z}_i)=0$ for $i=1,2,4$, the circle is the circumcircle of the face $f=(124)$, of center $\omega_f=\omega$ and radius $R_f=R$. The coefficient $a$ is given by $a=z_{41}\bar{z}_{21}-z_{21}\bar{z}_{41}$, which is the (euclidean) area of the face (124). Eventually we have:
\begin{equation}
D_{3,3}(T)-D_{3,3}(T')=\mathrm{Area}(f)\frac{|z_3-\omega_f|^2-R_f^2}{|z_{31}|^2|z_{32}|^2|z_{34}|^2}
\end{equation} 
which proves the lemma \ref{areadet}.

\section{Refined  lower bound for the volume: proof of theorem \ref{refined}}\label{appB}
The notations introduced here refer to the figure \ref{triang_base}. Each edge of the triangle $(abc)$ is surrounded by two faces. If we remove the point $z_{N+4}$, we obtain the Delaunay Triangulation for the points $\{z_1,\dots,z_{N+3}\}$, and the triangle $(abc)$ is one of its faces. Let us note $\theta_{(ab)}$, $\theta_{(bc)}$,and $\theta_{(ca)}$ the angles between the face $f=(abc)$ and the other face in contact with the edges $(ab)$, $(bc)$, and $(ca)$ respectively.\\
Now, in formula \ref{first_min}, instead of computing the integral of the measure over the region $\mathcal{B}(f)$, we carry out the integral over the region $\mathcal{R}(f)$. The integrand is not changed: it is the measure of the Delaunay triangulation made of the 4 points $a$, $b$, $c$ and $z_{N+4}$.
\begin{equation}
I_1 =\int_{\mathcal{R}(f)}d^2 z_{N+4} \det{\left[D_{u\bar{v}}(T^{D}(\{a,b,c,z_{N+4}\}))\right]_{u,v\notin\{ a,b,c\}}}.
\end{equation}
Then the computation of $I_1$ follows the same steps as for $I$, the only difference being the inequalities \ref{ineq} satisfied by $\theta_1$ and $\theta_2$:
\begin{align}
&\theta_1^{\mathrm{min}}-\frac{\theta_{(ca)}}{2}\leq  \theta_1  \leq \theta_1^{\mathrm{min}}+\frac{\pi}{2}&\\
&\theta_2^{\mathrm{min}}-\frac{\theta_{(ab)}}{2}\leq  \theta_2  \leq \theta_2^{\mathrm{min}}+\frac{\pi}{2}&\\
&\theta_1^{\mathrm{min}}+\theta_2^{\mathrm{min}}\leq  \theta_1 + \theta_2 \leq \theta_1^{\mathrm{min}}+\theta_2^{\mathrm{min}}+\frac{\theta_{(bc)}}{2} +\frac{\pi}{2}&
\end{align}
\begin{figure}
\centering
\includegraphics[width=13.5cm]{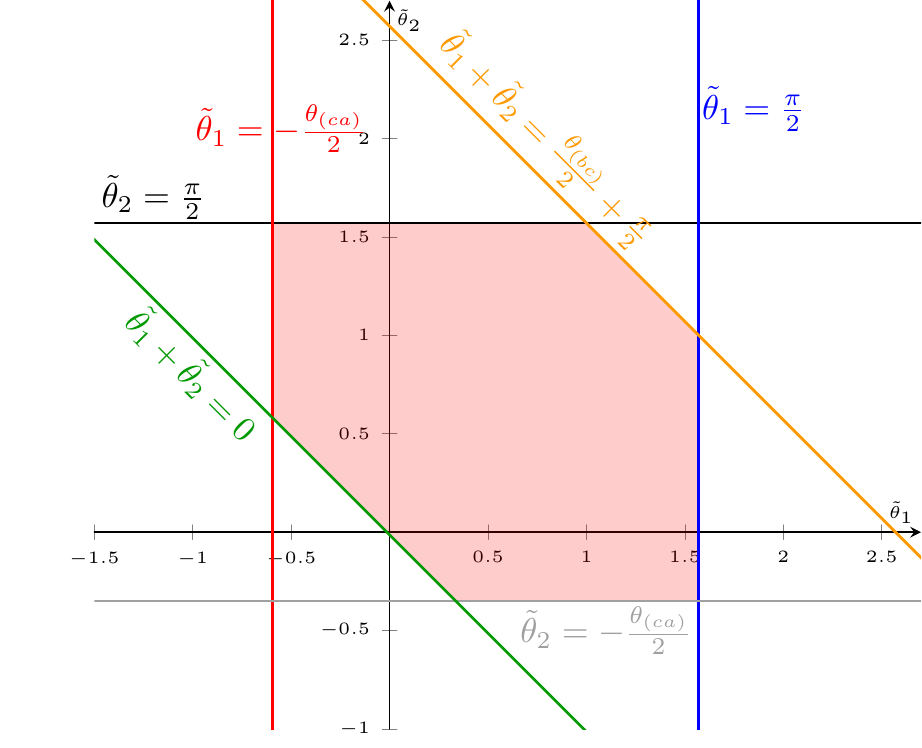}
\caption{Domain on which $d\tilde{\theta_1}d\tilde{\theta_2}$ has to be integrated. We take here $\tilde{\theta_i}=\theta_i-\theta_{i}^{\mathrm{min}}$.}
\label{repr_theta}
\end{figure}
The integral is then the area of the red region in figure \ref{repr_theta}.
So we get: $I_1=\frac{\pi^2}{16}+\frac{1}{16}[\theta_{(ab)}(2\pi-\theta_{(ab)})+\theta_{(bc)}(2\pi-\theta_{(bc)})+\theta_{(ca)}(2\pi-\theta_{(ca)})]$. Then, following the same steps as for the previous  lower bound, the result comes:
\begin{align}
    \int_{\mathbb{C}}d^2 z_{N+4} \, & \det{\left[D_{u\bar v}(T^{D}(\{z_1,\dots,z_{N+4}\}))\right]}_{u,v\notin\{ z_1,z_2,z_3\} }
    \nonumber\\
& \geq \left[(N+1)\frac{\pi^2}{8}+\frac{1}{8}\sum_{e\in\mathcal{E}}\theta_e (2\pi-\theta_e)\right]\det{\left[D_{u\bar v}(T^{D}(\{z_v\}))\right]}_{u,v\notin\{ z_1,z_2,z_3\}}
\end{align}

%

\end{document}